\documentclass{article}
\usepackage[margin=0.9in]{geometry} 
\usepackage{authblk}

\usepackage{graphicx}

\usepackage[utf8]{inputenc}
\usepackage[T1]{fontenc}

\usepackage{amsmath,amsfonts,theorem,epsfig,amssymb}
\usepackage{algpseudocode}
\usepackage{algorithm}
\usepackage{verbatim}

\newcommand{\real}{{\mathbb R}}

\newcommand{\natl}{{\mathbb N}}

\newcommand\sO{\mathcal{O}}
\newcommand\sU{\mathcal{U}}
\newcommand\sE{\mathcal{E}}

\newcommand\sY{\mathcal{Y}}

\newcommand{\bE}{{\mathbb{E}}}

\newcommand{\integ}{{\mathbb Z}}
\newcommand{\posint}{{\mathbb Z}_+}

\newtheorem{theorem}{Theorem}
\newtheorem{lemma}[theorem]{Lemma}

\newcommand{\qed}{\mbox{\rule{.4em}{1.7ex}\hspace{.6em}}}
\newenvironment{proof}{{\bf Proof\ }}{\hspace*{.1em}\hfill\qed
\bigskip \noindent}

\begin{document}


\title{State and parameter estimation from exact partial state observation in stochastic reaction networks}

\author[1]{Muruhan Rathinam}
\author[2]{Mingkai Yu}
\affil[1]{Department of Mathematics and Statistics, University of Maryland Baltimore County}
\affil[2]{Department of Mathematics and Statistics, University of Maryland Baltimore County}

\date{\today}

\maketitle

\begin{abstract}
We consider chemical reaction networks modeled by a discrete state and continuous in time Markov process for the vector copy number of the species and provide a novel particle filter method for state and parameter estimation based on exact observation of some of the species in continuous time. The conditional probability distribution of the unobserved states is shown to satisfy a system of differential equations with jumps. 
We provide a method of simulating a process that is a proxy for the vector copy number of the unobserved species along with a weight. The resulting weighted Monte Carlo simulation is then used to compute the conditional probability distribution of the unobserved species. We also show how our algorithm can be adapted for a Bayesian estimation of parameters and for the estimation of a past state value based on observations up to a future time. 
\end{abstract}

\section{Introduction}
Chemical reaction networks occurring at the intracellular level
often have some or all molecular species present in low copy numbers.
As such, these are best modeled by a discrete state Markov process in
continuous time, where the state $Z(t) \in \posint^n$ of the process is the vector copy number
of the molecular species at time $t$.
This model and the corresponding stochastic simulation algorithm introduced by
Gillespie \cite{Gillespie77,gillespie1976general} assumed a well-stirred system where spatial
variations in molecular concentrations were considered negligible.
Later on when the spatial variations were considered important, this model
still proved useful in that its framework allowed for a spatial
compartmental model where transport of molecules across compartments could be
treated mathematically as reactions \cite{spatialSSA}.

In this paper, we consider the problem of estimating the state $Z(t)$
from exact and continuous observations of some of the states in time.
In other words, suppose that the reaction network consists of $n$ molecular species
where the vector copy number of $n_2$ of the molecular species may be
observed exactly as a function of time over the interval $[0,t]$. Based on this observation, we are
interested in estimating the vector copy number of the remaining $n_1=n-n_2$
species at time $t$. If we denote the vector copy number of the observable species by
$Y(t) \in \posint^{n_2}$ and that of the rest of the (unobservable) species by $X(t) \in
\posint^{n_1}$ (so that $Z(t)=(X(t),Y(t))$), then our goal is to compute the conditional probability
distribution 
\[
\pi(t,x) = P\{X(t)=x \, | \, Y(s)=y(s), \; 0 \leq s \leq t \}
\]
for $x \in \posint^{n_1}$.

{\em Stochastic filtering} methods provide a framework for addressing this type of problem, which is not restricted to reaction networks. The stochastic filtering methods usually involve generating recursive
updates in time so that the additional computation required in going from the knowledge of conditional probability at
a time instant $t$ to a future time instant $t+h$ only involves the
observations made during $(t,t+h]$ and the conditional probability computed
by time $t$.   
The widely known instance of stochastic filtering is the Kalman filter  \cite{kalman1960new}
which was concerned with a Gaussian process with a linear evolution
in discrete time. Later on, the Kalman-Bucy filter \cite{kalman1961new}
was introduced in the continuous time setting of linear stochastic
differential equations (SDEs). In these linear and Gaussian settings the conditional
probability distribution is also Gaussian, and hence its time evolution can be
reduced to studying the time evolution of the mean and the covariance.

In general, one may not expect the conditional probability to be Gaussian
and hence the methods are more complex. In the general
nonlinear setting for SDEs there are several methods and one may find \cite{bain2008fundamentals} as a general reference. 
In the setting of Markov processes in discrete time, one may find 
recursive filtering equations which describe the time evolution
of the conditional probability distribution \cite{fristedt2007filtering}.

For discrete and finite state Markov processes in continuous time, a rigorous
derivation of the evolution equation for the conditional probability distribution $\pi(t,x)$ 
was provided in \cite{confortola2013filtering}. These evolution equations
take the form of differential equations with jumps. 
These equations, known as
the {\em filtering equations}, are analogous to {\em Kolmogorov's forward equations} 
for the time evolution of the (unconditional) probability distribution
$P\{Z(t)=z\}$ of a continuous time and discrete state Markov process. In the case of reaction networks, Kolmogorov's forward equations are known as the {\em Chemical Master Equations} (CME). 

It is well known that when the number of states is infinite
or very large, the solution of the CME is not practical and Monte Carlo 
simulation is the more efficient method of choice. The well known Gillespie
algorithm \cite{Gillespie77, gillespie2007stochastic} and its variants\cite{gibson2000efficient,anderson2007modified} are the methods of choice for
exact realizations of sample paths. Likewise, when the state space for $X(t)$
is very large or infinite, direct numerical solution of the filtering
equations is not practical. 

As we shall describe in Section \ref{sec-filtering-eq}, the evolution equation for $\pi(t,x)$ is nonlinear, and hence 
it is not possible to regard $\pi(t,x)$ as the (unconditional) probability distribution of a Markov process. We shall define a nonnegative function $\rho(t,x)$ 
which satisfies a linear evolution equation and when normalized yields $\pi(t,x)$: 
\[
\pi(t,x) = \frac{\rho(t,x)}{\sum_{\tilde{x}} \rho(t,\tilde{x})}.
\]
We refer to $\rho(t,x)$ as the {\em unnormalized conditional distribution} and refer to the evolution equations \eqref{eq-rho-deriv} and \eqref{eq-rho-jump} for $\rho(t,x)$, that are defined in Section \ref{sec-filtering-eq},
as the {\em unnormalized filtering equations}. 
We shall show that
$\rho(t,x)$ equals the expected value of a function of a suitably defined 
Markov process $(V(t),w(t))$:
\[
\rho(t,x) = \bE\left(1_{\{x\}}(V(t))\, w(t) \right),
\]
and thereby enabling an unbiased Monte Carlo estimate of $\rho(t,x)$. 
Here, $V$ has $\posint^{n_1}$ as its state space (thus $V$ has the same state space as the unobserved species $X$) and $w$ is a nonnegative weight.

To our best knowledge, in the context of discrete state and continuous time Markov chains, a Monte Carlo algorithm for the computation of $\pi(t,x)$ based on an exact partial state trajectory $y(t)$ observed in continuous time, is not available in the literature. In this paper, we provide such a weighted Monte Carlo algorithm, also referred to as a {\em particle filter}, that is tailored to chemical reaction networks. Our algorithm provides 
an unbiased estimate of $\rho(t,x)$, the unnormalized 
conditional distribution. 
The algorithm simulates $N_s$ identically distributed copies of $(V,w)$ 
such that $\rho(t,x)$ is estimated by
\[
\hat{\rho}(t,x) = \frac{1}{N_s} \, \sum_{i=1}^{N_s} 1_{\{x\}}(V^{(i)}(t))\, w^{(i)}(t).
\]
Thus, $\pi(t,x)$ is estimated by 
\[
\hat{\pi}(t,x) = \frac{\sum_{i=1}^{N_s} 1_{\{x\}}(V^{(i)}(t))\, w^{(i)}(t)}{\sum_{i=1}^{N_s} w^{(i)}(t)}.
\]
Due to the division, $\hat{\pi}(t,x)$ is not an unbiased estimator of $\pi(t,x)$. 
In fact, Monte Carlo methods of filtering in most other contexts (discrete time Markov chains, nonlinear SDEs etc.) also obtain the conditional probabilities by estimating an unnormalized version and then normalizing. 

We describe some previous results for state and parameter estimation in the context of reaction networks. All these methods are for observations based on discrete time snapshots while our work addresses the case of continuous time observation.  
\cite{boys2008bayesian} considers Bayesian inference of parameters based on exact partial state observations in discrete time snapshots and proposes
a Markov Chain Monte Carlo (MCMC) method. \cite{golightly2006bayesian, golightly2011bayesian} use the {\em Chemical Langevin Equation} (CLE) as the underlying model and propose MCMC methods for state and parameter estimation based on partial state observations corrupted by additive Gaussian noise in discrete time snapshots. 
\cite{calderazzo2019filtering} considers the case where the reaction network is approximated by the {\em linear noise approximation}. The observations 
are assumed to be linear combination of the states corrupted by an additive Gaussian noise with an unknown variance and they propose an extended Kalman-Bucy filter. They also consider pure delays in the system. The recent work in  \cite{fang2020stochastic} considers the case where the reaction network may be approximated by an ODE model with jumps. A function of the state corrupted by an additive Gaussian noise is observed in discrete time snapshots. 
Rigorous results are derived to show that by solving the filtering problem for the reduced model based on observations of the original model, one can also obtain a good approximation of the estimation of the original model.   

The weighted Monte Carlo method provided in this paper is concerned with the computation of
the conditional distribution of the unobserved species 
based on exact observation of some species in continuous time. 
In reality, perfect observations are not possible. However, it is often reasonable to assume that the observation noise is discrete in nature and hence the noise could be incorporated into the model via the introduction of additional species and reactions. For instance, during fluorescence intensity observations, photons produced constitute a species and their production by a fluorescently tagged molecule constitutes a reaction channel. The assumption of continuous in time observation is less realistic. Nevertheless, if the sampling rate is high compared to reaction rates, then we expect this to provide a good approximation. Moreover, having an algorithm for the idealised baseline case of continuous in time observations is expected to provide valuable insights into how to analyse and develop techniques for the case of discrete time observations of a continuous time process. 

The rest of the paper is organized as follows. In Section
\ref{sec-filtering-eq} we briefly review the basics of stochastic reaction
networks and provide the filtering equations. We also introduce the useful
{\em unnormalized filtering equations}. In Section \ref{sec-monte-carlo} we
provide a Monte Carlo filtering algorithm which generates a pair of processes $(V,w)$ 
where $V$ has the same state space as $X$ such that 
\[
\pi(t,x) = \frac{\bE[1_{\{x\}}(V(t))\, w(t)]}{\bE[w(t)]}.
\]
We also 
show how our algorithm can be easily adapted to do parameter estimation
based on exact partial state observation in a Bayesian framework and also how to estimate 
the conditional probability of a past event based on 
the observations made until a later time. That is, we show how to compute  
\[
P\{X(t)=x \, | \, Y(s)=y(s), \, 0 \leq s \leq T \},
\]
where $x \in \posint^{n_1}$ and $0<t<T$, via an easy modification of our algorithm.
In Section \ref{sec-examples} we illustrate our algorithms via numerical examples and Section \ref{sec-conclusions} makes some concluding remarks and discusses future work.

\section{Stochastic reaction networks and filtering equations}\label{sec-filtering-eq} 

We consider a stochastic reaction network with $n$ molecular species and $m$ reaction
channels. The process $Z(t) \in \posint^n$ stands for the copy number vector
of the species at time $t$. The dynamics of the reaction network are
characterized by {\em propensity functions}
$a_j(z,c)$ where $z \in \posint^n$ is the state and $c \in \real^p$ is a vector of parameters, and also by the stoichiometric vectors $\nu_j \in \integ^n$ for $j=1,\dots,m$.
The occurrence of a reaction event $j$ at time $t$ leads to a state change
$Z(t)=Z(t-) + \nu_j$ and given $Z(t)=z$, the probability that a reaction event
$j$ occurs during $(t,t+h]$ is $a_j(z,c) h + o(h)$ as $h \to 0+$.
We shall adopt the convention that the process $Z(t)$ is right continuous
and assume $Z(t)$ to be non-explosive, meaning that there are only finitely
many reaction events during any finite time interval.  
For brevity, we suppress the dependence on parameters except when we
are concerned with parameter estimation. 
It is well known\cite{Gillespie77, gillespie2007stochastic, gillespie1976general, gillespie2007stochastic} that the time evolution of the (unconditional) probability mass function
\[
p(t,z) = P\{Z(t)=z\}
\]
is given by the {\em Kolmogorov's forward equations} also known as the {\em
  chemical master equations} (CME):
\begin{equation}
p'(t,z) = \sum_{j=1}^m a_j(z-\nu_j) \, p(t,z-\nu_j) - \sum_{j=1}^m a_j(z)
\, p(t,z).
\end{equation}

We consider the situation where we make exact (noiseless) observations of the copy number of
the last $n_2$ species in continuous time. We write $Z(t)=(X(t),Y(t))$ where
$X(t) \in \posint^{n_1}$ is the unobserved component of the state and
$Y(t) \in \posint^{n_2}$ is the observed component of the state.
Suppose we make a particular observation $Y(s)=y(s)$ for $0 \leq s \leq t$ for
$t \geq 0$.
We are interested in computing the {\em conditional probability}
\begin{equation}
  \pi(t,x) = P\{ X(t)=x \, | \, Y(s)=y(s), \, 0 \leq s \leq t \} \quad \forall x \in \posint^{n_1}.
\end{equation}

Let us denote by $\nu'_j$ the first $n_1$ components and by $\nu''_j$ the last $n_2$ components of $\nu_j$ for
$j=1,\dots,m$ and define the subset $\sO \subset \{1,\dots,m\}$ consisting
of {\em observable reaction channels}, that is those that alter $Y(t)$. Thus $j \in \sO$ if and only if
$\nu''_j \neq 0$. We denote by $\sU$, the complement of $\sO$. Thus $\sU$
consists of the {\em unobservable reaction channels}.

We shall denote by $t_k, k=1,2,\dots$ the successive jump times of $y(t)$ and let
$t_0=0$. 
Let's define $a^{\sO}(x,y)$ and $a^{\sU}(x,y)$ by
\begin{equation}
a^{\sO}(x,y) = \sum_{j \in \sO}a_j(x,y), \quad  a^{\sU}(x,y) = \sum_{j \in
  \sU}a_j(x,y),
\end{equation}
which are respectively the total propensity of the observable and the
unobservable reactions. For each $k \in \natl$ define $\sO_k$ by
\[
\sO_k = \{j \in \sO\, | \, \nu''_j=y(t_k)-y(t_k-)\}.
\]
Thus $\sO_k$ is the subset of observable reactions
that are consistent with the observed jump $y(t_k)-y(t_k-)=y(t_k)-y(t_{k-1})$ at time $t_k$. 

The time evolution of $\pi(t,x)$ follows a system of differential equations
in between jump times $t_k$ and at the jump times $\pi(t,x)$
undergoes jumps. A rigorous derivation of the evolution equations for $\pi(t,x)$ (for finite state Markov processes) is given in \cite{confortola2013filtering} 
and we provide a more intuitive derivation in the Appendix. We summarize these equations which we call the {\em filtering
equations} here.

For $k=0,1,\dots$, and for $t_k \leq t < t_{k+1}$, $\pi$ satisfies the following system of differential equations: 
\begin{equation}\label{eq-pi-deriv}
\begin{aligned}  
\pi'(t,x) &= \sum_{j \in \sU} \pi(t,x-\nu_j')\,a_j(x-\nu_j',y(t_k)) - \sum_{j \in
  \sU} \pi(t,x)\,a_j(x,y(t_k))\\
&- \pi(t,x) \left(a^{\sO}(x,y(t_k)) -
\sum_{\tilde{x}}a^{\sO}(\tilde{x},y(t_k)) \pi(t,\tilde{x}) \right) \quad \forall x \in \posint^{n_1}
\end{aligned}
\end{equation}
and for $k=1,2,\dots$, and at times $t_k$, $\pi(t,x)$ jumps according to:
\begin{equation}\label{eq-pi-jump}
\pi(t_k,x) = \frac{\sum_{l \in \sO_k} a_l(x-\nu_l',y(t_{k-1}))\,
  \pi(t_k-,x-\nu_l')}{\sum_{\tilde{x}}\sum_{l \in \sO_k}
  a_l(\tilde{x},y(t_{k-1}))\, \pi(t_k-,\tilde{x})} \quad \forall x \in \posint^{n_1}.
\end{equation}
We note that $\pi(t,x)$ is right continuous since by our convention
the process $Z(t)$ is right continuous and as a consequence the observed
trajectory $y(t)$ is right continuous as well.

Since we are interested in the case where the state space of the system of ODEs \eqref{eq-pi-deriv} is either an
infinite or a very large finite subset of $\posint^{n_1}$, a direct solution
of this equation is not practical. Instead our goal is
a Monte Carlo simulation.
Our Monte Carlo algorithm involves generating a trajectory $V(t)$ with
same state space $\posint^{n_1}$ as $X(t)$, along
with a weight trajectory $w(t)$, 
such that at any time $t \geq 0$ and for $x \in \posint^{n_1}$ 
\begin{equation}\label{eq-pi-MonteCarlo}
\pi(t,x) = \frac{\bE [w(t) 1_{\{x\}}(V(t))]}{\bE[w(t)]}.
\end{equation}
We note that, given a set $A$, $1_A$ is the indicator function of the set.  
In practice, from a sample of $N_s$ identically distributed trajectories $V^{(i)}$ along with
weights $w^{(i)}$ for $i=1,\dots,N_s$, $\pi(t,x)$ is estimated by
\begin{equation}\label{eq-pi-hat}
\hat{\pi}(t,x) = \frac{\sum_{i=1}^{N_s} 1_{\{x\}}(V^{(i)}(t)) \,
  w^{(i)}(t)}{\sum_{i=1}^{N_s} w^{(i)}(t)}
\end{equation}

Since \eqref{eq-pi-deriv} is not linear, 
$\pi$ cannot be interpreted as the probability mass function of
some Markov process. However, we can define a related quantity
$\rho(t,x)$ which evolves according to a linear equation and can be related to
a Markov process. We define $\rho(t,x)$ to be a nonnegative function that satisfies the {\em unnormalized
filtering equations} defined as below. 

For $k=0,1,\dots$ on the interval $t_k \leq t < t_{k+1}$, $\rho(x,t)$ satisfies 
\begin{equation}\label{eq-rho-deriv}
\begin{aligned}  
\rho'(t,x) &= \sum_{j \in \sU} \rho(t,x-\nu_j')\,a_j(x-\nu_j',y(t_k)) - \sum_{j \in
  \sU} \rho(t,x)\,a_j(x,y(t_k))\\
&- \rho(t,x)\, a^{\sO}(x,y(t_k)) \quad \forall x \in \posint^{n_1}.
\end{aligned}
\end{equation}
For $k=1,2,\dots$ at jump times $t_k$, $\rho(t,x)$ jumps according to 
\begin{equation}\label{eq-rho-jump}
\rho(t_k,x) = \frac{1}{|\sO_k|}\sum_{j \in \sO_k} a_j(x-\nu_j',y(t_{k-1})) \, \rho(t_k-,x-\nu_j') \quad x \in \posint^{n_1}.      
\end{equation}
We note that $|\sO_k|$ is the number of reaction channels in the set $\sO_k$ and this is non-zero.
Let $\bar{\rho}(t)=\sum_{\tilde{x}}\rho(t,\tilde{x})$.
It is shown in Appendix \ref{append-rho-pi} that $\pi$ is given by  $\pi(t,x)=\rho(t,x)/\bar{\rho}(t)$. 
In the next section we shall see that an appropriate stochastic interpretation
of $\rho(t,x)$ leads to the stochastic simulation algorithm to generate $V$ and $w$ such that 
\begin{equation}\label{eq-stoch-rho}
\bE\left[1_{\{x\}}(V(t)) \, w(t) \right] = \rho(t,x).   
\end{equation}

\section{Monte Carlo Algorithms}\label{sec-monte-carlo}

\subsection{Weighted Monte Carlo and Resampling}
In Section \ref{sec-Vw} we shall define a pair of processes $(V,w)$ with state space $\posint^{n_1} \times [0,\infty)$ such that \eqref{eq-stoch-rho} holds where $\rho$ satisfies the unnormalized filtering equations \eqref{eq-rho-deriv} and \eqref{eq-rho-jump}. 
The estimation of $\pi(t,x)$ is accomplished by simulating $N_s$ identically distributed copies of the pair of processes $(V,w)$ and estimating $\pi(t,x)$ via \eqref{eq-pi-hat}. 

We note that this type of procedure differs from the most common form of 
simulating a reaction network via $N_s$ i.i.d.\ trajectories, say $Z^{(i)}$ and then estimating 
$\bE[\varphi(Z(t))]$, the expectation of a function $\varphi$ of the state via the sample average
\[
\frac{1}{N_s} \sum_{i=1}^{N_s} \varphi(Z^{(i)}(t)).
\]
In this most familiar form of Monte Carlo simulation, all trajectories are weighted equally. In our situation, the Monte Carlo algorithm involves a weighted average. The weights carry information about the relative importance of trajectories of $V$.

Corresponding to the sample $(V^{(i)},w^{(i)})$ where $i=1,\dots,N_s$, we may associate the (random) empirical measure (mass) $M(t)$ at time $t$ given by 
\begin{equation}\label{eq-emp-mass}
M(t) = \frac{1}{N_s}\sum_{i=1}^{N_s} w^{(i)}(t) \, \delta_{V^{(i)}(t)},
\end{equation}
where $\delta_x$ stands for the Dirac mass or unit point mass concentrated at $x \in \posint^{n_1}$.
If we choose $(V,w)$ to satisfy \eqref{eq-stoch-rho}, then for each $x \in \posint^{n_1}$ the expected value of the empirical measure of the singleton $\{x\}$ is $\rho(t,x)$:
\begin{equation}
\bE[M(t) (\{x\}) ] = \rho(t,x).    
\end{equation}

The fact that the weights may grow (or decay) unevenly as time progresses leads to two different issues. The first issue is that the simulation may run into numerical problems where some weights become very large while others become very small, exceeding the finite precision of the computer. Especially, very large weights will become {\tt Inf} in the finite precision representation of the computer, rendering computations impractical. The second issue is more fundamental and is unrelated to the finite precision nature of computations. The accuracy of the Monte Carlo estimate of $\pi(t,x)$ depends on the variance of $w(t)$ 
as well as the variance of $1_{\{x\}}(V(t)) \, w(t)$ and also the covariance between these two terms. If the algorithm allows for a large variance in $w(t)$, naturally this would have negative implications for the accuracy of the estimate. 

In order to avoid these issues, one may {\em resample} the empirical measure at 
certain time points to produce a new empirical measure that is still the sum of $N_s$ Dirac masses but with equal weights. By resampling, we mean the creation a new weighted sample of size $N_s$ 
from an existing weighted sample of size $N_s$
that captures the information in the original sample.   
In particular, if an empirical measure $M(t)$ at time $t$ is resampled to produce another empirical measure $\tilde{M}(t)$, then we want the resampling procedure to satisfy the basic condition that 
\begin{equation}
\bE[\tilde{M}(t) (\{x\}) ] =   \bE[M(t) (\{x\}) ] \quad \forall x \in \posint^{n_1}.  
\end{equation}
In order to maintain above requirement, the resampling procedure itself introduces extra randomness which in turn can result in loss of accuracy. Thus the frequency with which resampling is performed will be an important topic of investigation. We also note that the resampling procedure introduces dependence among the sample trajectories. Thus, while $(V^{(i)},w^{(i)})$ for $i=1,\dots,N_s$ are identically distributed, the collection is not necessarily independent.

Suppose that at time $t_0>0$, we have simulated pairs $(V^{(i)}(t_0),w^{(i)}(t_0))$ for $i=1,\dots,N_s$, and we resample to produce a new set of pairs 
$(\tilde{V}^{(i)}(t_0),\tilde{w}^{(i)}(t_0))$ for $i=1,\dots,N_s$. Then we shall have that for all $i$ and $j$, 
\[
\bE[1_{\{x\}}(\tilde{V}^{(j)}(t_0)) \, \tilde{w}^{(j)}(t_0)] = \bE[1_{\{x\}}(V^{(i)}(t_0)) \, w^{(i)}(t_0)],
\]
and $\tilde{w}^{(j)}(t_0)=1$ for $j =1,\dots,N_s$.
The consequent evolution of $(\tilde{V}^{(j)}(t),\tilde{w}^{(j)}(t))$ for $j=1,\dots,N_s$ (which we shall rename $(V^{(j)},w^{(j)})$), will proceed in a conditionally independent fashion until the next resampling. 

We shall employ a particular resampling algorithm described 
in \cite{bain2008fundamentals} and ~\cite{crisan2002minimal}. This is given in Algorithm \ref{alg-offspring} and involves regarding each sample point/particle as giving birth to a number (possibly zero) of offsprings. 
It is most convenient to think of this algorithm as propagating $N_s$ particles so that at time $t$, the $i$th particle is situated at location $V^{(i)}(t) \in \posint^{n_1}$ and 
has weight $w^{(i)}(t) \geq 0$. When resampling occurs at say, time $t_0$, Algorithm \ref{alg-offspring} considers the weight of each particle $i$ and assigns a random number $o^{(i)} \in \posint$ of offsprings to the particle. The number $o^{(i)}$ is chosen such that $\bE[o^{(i)} \, | \, w^{(i)}(t_0)]=w^{(i)}(t_0)$ and subject to the condition that the total number of offsprings of all particles is $N_s$. Each particle $i$ is considered ``dead'' after it gives ``birth'' to $o^{(i)}$ offsprings each of which start at the same location $V^{(i)}(t_0)$ as its parent and with weight $w^{(i)}(t_0)=1$. Subsequently all the offsprings evolve independently. We refer the reader to \cite{bain2008fundamentals} (Section 9.2) and \cite{crisan2002minimal} regarding further details of a specific version of this algorithm that is also used by us and described in Algorithm \ref{alg-offspring}.  

\subsection{Processes $V$ and $w$}\label{sec-Vw}

In this section we define a pair of processes $(V,w)$ where $V(t) \in \posint^{n_1}$ and $w(t) \geq 0$ as follows. 
We initialize $V(0)$ with the same distribution as that of the initial state $X(0)=x_0$ and set $w(0)=1$. 
In between jump times, that is, for $t_k \leq t < t_{k+1}$ where $k=0,1,\dots$, the process $V$ is evolved according to the underlying reaction network with {\em all observable reactions removed}. Thus only 
the copy number of the first $n_1$ species can change during $(t_k,t_{k+1})$. Moreover, for $t_k \leq t < t_{k+1}$ where $k=0,1,\dots$, the weight process $w(t)$ is evolved according to 
\begin{equation}\label{eq-w}
w(t) = w(t_{k}) \,\exp\left\{-\int_{t_k}^t a^{\sO}(V(s),y(t_{k}))\, ds \right\}.
\end{equation}
Equivalently, we note that $w(t)$ satisfies the ODE
\[
w'(t) = w(t) \, a^{\sO}(V(t),y(t_k)) 
\]
for $t_k \leq t < t_{k+1}$ (at $t=t_k$ the right-hand side derivative is considered).
At jump times $t=t_k$ for $k=1,2,\dots$ the process $w(t)$ jumps as follows. For $j \in \sO_k$, set
\begin{equation}\label{eq-w-jump}
w(t_k) = w(t_k-) \, a_j(V(t_k-),y(t_{k-1})) \;\; \text{ with probability } \frac{1}{|\sO_k|}.
\end{equation}
Moreover, at jump times $t=t_k$ for $k=1,2,\dots$ the process $V(t)$ jumps as follows. For $j \in \sO_k$, set
\begin{equation}\label{eq-V-jump}
V(t_k) = V(t_k-) + \nu_j' \;\;  \text{ with probability } \frac{1}{|\sO_k|}.  
\end{equation}
We note that if $a_j(V(t_k-),y(t_{k-1}))=0$ for the chosen 
$j \in \sO_k$, it may be possible that $V(t_k)$ is assigned an infeasible state, that is a state that has negative components. However, at the same time $w(t_k)$ would be zero and hence the value of $V(t_k)$ would not matter.

Our goal is to show that 
\[
\bE[1_{\{x\}}(V(t)) \ w(t) ] = \rho(t,x)
\]
for all $t$ and $x$, where $\rho$ is defined by \eqref{eq-rho-deriv} and \eqref{eq-rho-jump}.
The following two lemmas basically establish that. 
\begin{lemma}\label{lem-feynman-kac}
Let $V$ and $w$ be defined as above, and let $\rho$ defined by
\begin{equation}\label{eq-rho-interp2}
  \rho(t,x) = \bE[ 1_{\{x\}}(V(t)) \, w(t)].
\end{equation}
If $V$ is non-explosive then $\rho(t,x)$ satisfies \eqref{eq-rho-deriv}
 on $t_k \leq t < t_{k+1}$. 
\end{lemma}
\begin{proof}
As a function of $t$, $1_{\{x\}}(V(t))$  is of bounded variation and
piecewise constant on every bounded interval of time. Moreover, $w(t)$ is absolutely
continuous in $t$. Thus we may write for $t_k \leq t < t_{k+1}$
\[
\begin{aligned}
1_{\{x\}}(V(t)) \, w(t) &= 1_{\{x\}}(V(t_k)) \, w(t_k) + \int_{t_k}^t 1_{\{x\}}(V(s))
\, w'(s) ds\\
&+ \sum_{t_k < s \leq t} w(s) \, \left(1_{\{x\}}(V(s)) -
  1_{\{x\}}(V(s-))\right)\\
&=1_{\{x\}}(V(t_k)) \, w(t_k) - \int_{t_k}^t 1_{\{x\}}(V(s)) \, w(s) \, a^{\sO}(V(s),y(t_k))
  \, ds \\
&+ \sum_{j \in \sU} \sum_{t_k <s \leq t} w(s) \, \left(1_{\{x\}}(V(s-)+\nu_j') -
  1_{\{x\}}(V(s-))\right) \, \left(R_j(s)-R_j(s-)\right),
\end{aligned}
\]
where $R_j$ is the process that counts the number of
firings of reaction channel $j$ during $(0,t]$. 
Since the stochastic intensity \cite{Bremaud, anderson2011continuous} of $R_j$ is given by 
$a_j(V(t-))$, taking expectations we obtain
\[  
\begin{aligned}
\rho(t,x) &= \rho(t_k,x) + \sum_{j \in \sU} \bE\left[\int_0^t \left(1_{\{x\}}(V(s)+\nu_j') -
  1_{\{x\}}(V(s))\right) \, w(s) \, a_j(V(s),y(t_k)) \, ds\right]\\
&- \bE\left[ \int_0^t 1_{\{x\}}(V(s)) \, w(s) \, a^{\sO}(V(s),y(t_k))
  \, ds\right].
\end{aligned}
\]  
Using Fubini and differentiating with respect to $t$, We obtain that
\[
\begin{aligned}
\rho'(t,x) &= \sum_{j \in \sU} \bE[ (1_{\{x\}}(V(t)+\nu_j')-1_{\{x\}}(V(t))) \,
  a_j(V(t),y(t_k)) \, w(t)]\\
&- \bE[ 1_{\{x\}}(V(t))\, w(t)\,
  a^{\sO}(V(t),y(t_k))]\\
&= \sum_{j \in \sU} \bE[1_{\{x-\nu_j'\}}(V(t)) \, a_j(x-\nu_j',y(t_k))\, w(t)]
- \sum_{j \in \sU} \bE[1_{\{x\}}(V(t)))\, a_j(x,\nu_j') \, w(t)]\\
&- \bE[ 1_{\{x\}}(V(t))\, w(t)\,
  a^{\sO}(x,y(t_k))]\\
&= \sum_{j \in \sU} \rho(t,x-\nu_j') \, a_j(x-\nu_j',y(t_k)) - \sum_{j \in
  \sU} \rho(t,x) \, a_j(x,y(t_k)) - \rho(t,x) \, a^{\sO}(t,x),
\end{aligned}
\]
which agrees with \eqref{eq-rho-deriv}. 
\end{proof}

\begin{lemma}
Let $k \in \{1,2,\dots\}$ and suppose 
\[
\bE[ 1_{\{x\}}(V(t_k-)) \, w(t_k-) ] = \rho(t_k-,x)
\]
for all $x$. Then 
\[
\bE[ 1_{\{x\}}(V(t_k)) \, w(t_k) ] = \rho(t_k,x)
\]
\end{lemma}
\begin{proof}
It follows that the conditional expectation
\[
\begin{aligned}
\bE[ &w(t_k)\, 1_{\{x\}}(V(t_k)) \, | \, w(t_k-), V(t_k-) ] \\
 &= \frac{1}{|\sO_k|} \, \sum_{j \in \sO_k} w(t_k-) \, a_j(x-\nu_j',y(t_{k-1})) \,
 1_{\{x-\nu_j'\}}(V(t_k-)).
\end{aligned}
\]
Taking expectation, we obtain
\[
\bE[ w(t_k)\, 1_{\{x\}}(V(t_k))] = \frac{1}{|\sO_k|} \, \sum_{j \in \sO_k} \rho(t_k-,x-\nu_j') \, a_j(x-\nu_j',y(t_{k-1})).
\]
The result follows from \eqref{eq-rho-jump}.
\end{proof}

From the above lemmas and mathematical induction, it follows that 
\[
\bE[1_{\{x\}}(V(t)) \, w(t)] = \rho(t,x),
\]
for all $t \geq 0$ and $x$.

Algorithm \ref{alg-overall} describes the overall filtering algorithm whereas Algorithm \ref{alg-continuous-nokill} describes the continuous evolution and Algorithm \ref{alg-jump} describes the jumps at observed jump times $t_k$. Algorithm \ref{alg-offspring} describes the resampling via offsprings \cite{bain2008fundamentals}. We observe that in Algorithm \ref{alg-overall}, resampling is only performed at the observed jump times $t_k$ and not necessarily always. In Section \ref{sec-resampling} we numerically explore different resampling strategies. One extreme option is always to resample at $t_k$, the other extreme is never to resample. The third option is to resample at $t_k$ if either the number of zero weights among the $N_s$ particles exceeds a predetermined number or if the ratio of the largest weight to the smallest non-zero weight 
among the particles exceeds a predetermined value. 
If resampling is not performed at $t_k$, then the 
weights are rescaled so that the average weight is $1$. 
This latter step prevents all weights from growing 
to be too large or too small. We also note that it is possible to consider resampling at time points which may not coincide with the observed jump times $t_k$. However, since the only point in time where some weights may become zero is at the jump times $t_k$, it makes sense to consider the jump times as the points in time to determine if a resampling step is required.   

\begin{algorithm}[th]  
\caption{Overall scheme}
\label{alg-overall} 
\begin{algorithmic}[1]
\State \textbf{Input:} Jump times of $Y$ $(t_1,\dots,t_{N_k})$ and observed $Y$ at jump times $(y_1,\dots,y_{N_k})$. Initial distribution $\mu_0$ for $X(0)=x_0$, parameter value $c$, final time $T$, filter sample size $N_s$.
\State 
Generate $N_s$ i.i.d\ sample from $\mu_0$ and assign to $V^{(1)},\dots,V^{(N_s)}$. Set $w^{(i)}=1$ for $i=1,\dots,N_s$.
\State $k=1$, $t=0$, $y=y_0$ 
\For { $k = 1$ to $N_k$}
\For { $i=1$ to $N_s$ }
\State $(V^{(i)}_-,w^{(i)}_-) = \text{Continuous-evolution}(V^{(i)},w^{(i)},t,t_k,y,c)$
\State $(V^{(i)},w^{(i)}) = \text{Jump}(V^{(i)}_-,w^{(i)}_-,y,y_k,c)$
\EndFor
\If {resampling}  $(V, w)=\text{Offsprings}(V, w)$
\Else { rescale } $w$ so that $\sum_{i=1}^{N_s}w^{(i)}=N_s$
\EndIf
\State Set $t=t_k$, $y=y_k$
\EndFor
\For { $i=1$ to $N_s$}
\State $(V^{(i)},w^{(i)}) =
\text{Continuous-evolution}(V^{(i)},w^{(i)},t_{N_k},T,y_{N_k},c)$
\EndFor
\end{algorithmic}
\end{algorithm}

\begin{algorithm}[ht]                       
\caption{Continuous evolution}
\label{alg-continuous-nokill}  
\begin{algorithmic}[1]
  \Function{Continuous-evolution}{$V_0,w_0,t_0,t_f,y, c$}
  \State Assume unobservable reactions are numbered $1,2,\dots,m_u$
  \State Assume observable reactions are numbered $m_u+1,\dots,m$
  \State Set $V=V_0$, $w=w_0$, $t=t_0$
  \While{$t < t_f$}
  \For {$j=1$ to $m_u$}
  \State Set $\lambda_j=a_j(V,y,c)$
  \EndFor  
  \State Set $\lambda^{\sO}=\sum_{j=m_u+1}^m \lambda_j$
  \State Set $\lambda^{\sU}=\sum_{j=1}^{m_u} \lambda_j$
  \State Generate $\tau \sim \text{Exponential}(\lambda^{\sU})$
  \If{$t+\tau < t_f$}
  \State Generate $u \sim \text{Uniform}[0,1]$
  \State Find $j \in \{1,\dots,m_u\}$  such that $\sum_{l=1}^{j-1} \lambda_l < u\lambda^{\sU} \leq  \sum_{l=1}^{j} \lambda_l$
  \State Set $V=V+\nu_j'$
  \EndIf 
  \State Set $t_n=\min\{t+\tau,t_f\}$
  \State Set $w = w \times \exp\{-(t_n-t) \, \lambda^{\sO}\}$
  \State Set $t = t_n$
  \EndWhile
  \State \Return $(V,w)$
  \EndFunction
\end{algorithmic}
\end{algorithm}

\begin{algorithm}[th]                       
\caption{Jump}
\label{alg-jump}  
\begin{algorithmic}[1]
  \Function{Jump}{$V_-,w_-,y_-,y,c$}
  \State Assume unobservable reactions are numbered $1,2,\dots,m_u$
  \State Assume observable reactions are numbered $m_u+1,\dots,m$
  \State Find indices $j \in \{m_u+1,\dots,m\}$ that satisfy $\nu_j''=y-y_-$.
  \State Let above indices be stored in array $J$ of length $L$. 
  \For {$i=1$ to $L$}
  \State Set $\lambda_i=a_{J(i)}(V_-,y_-,c)$
  \EndFor
  \If {$L>1$}
  \State Generate $u \sim \text{Uniform}[0,1]$
  \State Let $i \in \{1,\dots,L\}$ be such that $i-1< u \, L \leq i$
  \Else
  \State Set $i=1$
  \EndIf
  \State Set $V=V_- +\nu_{J(i)}'$
  \State Set $w = w_- \, \lambda_{J(i)}$
  \State \Return $(V,w)$
  \EndFunction
\end{algorithmic}
\end{algorithm}

\begin{algorithm}[th]
\caption{Offsprings}
\label{alg-offspring} 
\begin{algorithmic}[1]
\Function{Offsprings}{$V,w$}
\State We note that $V=(V_1,\dots,V_{N_s})$ and $w=(w_{1},\dots,w_{N_s})$.  We denote $[x]$ by the integer part of a number $x$, that is the greatest integer less than or equal to $x$, while $\{x\} = x - [x] $ the fractional part of $x$.
\State Normalize $\bar{w}_i = w_i/\sum_{i=1}^N w_i$.
\State Simulate i.i.d random variables $u_i$ $\sim$ Unif[0,1], for $i=1,...,N-1$
\State Initialize $g=N, h=N$
\For{$i=1$ to $N-1$}
\If{$\{N\bar{w}_i\} + \{g- N\bar{w}_i\} < 1$}
\If{$u_i < 1 - (\{N\bar{w}_o\}/\{g\})$}
\State $o_i = [N\bar{w}_i]$
\Else
\State $o_i = [N\bar{w}_i] + (h- [g]) $
\EndIf
\Else
\If{$u_i < 1 - \frac{1-\{N\bar{w}_i\}}{1-\{g\}}$}
\State $o_i =  [N\bar{w}_i] + 1$
\Else
\State $o_i =  [N\bar{w}_i] + (h-[g])$
\EndIf
\EndIf
\EndFor
\State Set $j=1$
\For{$i=1$ to $N_s$}
\State $w_i=1$
\For{$l=1$ to $o_i$}
\State $\tilde{V}_j = V_i$, $j = j + 1$
\EndFor
\EndFor
\State \Return $(\tilde{V},w)$
\EndFunction
\end{algorithmic}
\end{algorithm}

\subsection{Treating parameters as random variables}
In general, the propensity functions $a_j(x,y)$ depend on a vector $c$ of
parameters, thus $a_j=a_j(x,y,c)$ where $c \in \real^p$. Here we consider
a Bayesian framework for inferring the parameters $c$ from the observation
$Y(s)=y(s)$ for $0 \leq s \leq t$. This involves treating $c$ as a random
variable, or rather a stochastic process $C(t)$ which remains constant in
time $t \geq 0$. Thus, we may ``absorb'' $C(t)$ into $X(t)$, thus expanding
the unobserved components of the state by $p$ extra dimensions. 
The Bayesian framework involves starting with a prior distribution
$\bar{\mu}$ on the parameter space $\real^p$. Then we compute the posterior
distribution probability mass/density function $\pi(t,x,c)$ which is
characterized by 
\[
P\{X(t)=x, C(t) \in A \,| \, Y(s)=y(s), \, 0 \leq s \leq t\} = \int_{c \in A}
\pi(t,x,c) \, dc.
\]
Since the parameters are distributed continuously, the filtering equations \eqref{eq-pi-deriv} and \eqref{eq-pi-jump}
need to be modified as follows.
For $t_k \leq t < t_{k+1}$, $\pi$ satisfies the following system of differential equations: 
\begin{equation}\label{eq-pi-deriv-para}
\begin{aligned}  
\pi'(t,x,c) &= \sum_{j \in \sU} \pi(t,x-\nu_j',c)\,a_j(x-\nu_j',y(t_k),c) - \sum_{j \in
  \sU} \pi(t,x,c)\,a_j(x,y(t_k),c)\\
&- \pi(t,x,c) \left(a^{\sO}(x,y(t_k),c) -
\int_{\tilde{c}} \sum_{\tilde{x}}a^{\sO}(\tilde{x},y(t_k),\tilde{c}) \pi(t,\tilde{x},\tilde{c}) d\tilde{c}\right),
\end{aligned}
\end{equation}
which holds for all $x \in \posint^{n_1}$ and $c \in \real^p$. 
At times $t_k$, $\pi(t,x,c)$ jumps according to:
\begin{equation}\label{eq-pi-jump-para}
\pi(t_k,x,c) = \frac{\sum_{l \in \sO_k} a_l(x-\nu_l',y(t_{k-1}),c)\,
  \pi(t_k-,x-\nu_l',c)}{\int_{\tilde{c}} \sum_{\tilde{x}}\sum_{l \in \sO_k}
  a_l(\tilde{x},y(t_{k-1}),\tilde{c})\, \pi(t_k-,\tilde{x},\tilde{c}) d\tilde{c}},
\end{equation}
which holds for all $x \in \posint^{n_1}$ and $c \in \real^p$. 

Algorithm \ref{alg-overall-bayes} describes the Monte Carlo algorithm to
generate a weighted sample from the posterior distribution for the parameter. 
Having computed $C^{(i)}$ and $w^{(i)}$ for $=1,\dots,N_s$, we may estimate
the posterior mean $\bar{C}$ and standard deviation $\sigma_C$ by
\begin{equation}\label{eq-C-bar}
  \bar{C} = \frac{\sum_{i=1}^{N_s} w^{(i)} \, C^{(i)}}{\sum_{i=1}^{N_s} w^{(i)}},
\end{equation}
and
\begin{equation}
  \sigma_C^2 = \frac{\sum_{i=1}^{N_s} w^{(i)} \, (C^{(i)}-\bar{C})^2}{\sum_{i=1}^{N_s} w^{(i)}}.
\end{equation}

\begin{algorithm}[th]  
\caption{Overall scheme for Bayesian inference}
\label{alg-overall-bayes} 
\begin{algorithmic}[1]
\State \textbf{Input:} Jump times of $Y$ $(t_1,\dots,t_{N_k})$ and observed $Y$ at jump times $(y_1,\dots,y_{N_k})$. Initial distribution $\mu_0$ for $X(0)=x_0$, prior distribution $\bar{\mu}$ for parameter $c$, final time $T$, filter sample size $N_s$.  
\State Generate a sample $x_0$ from $\mu_0$ and a sample $c$ from $\bar{\mu}$. Set $V^{(i)}=x_0$, $w^{(i)}=1$, $C^{(i)}=c$ for $i=1$ to $N_s$
\State Set $k=1$, $t=0$, $y=y_0$ 
\For { $k = 1$ to $N_k$}
\For { $i=1$ to $N_s$ }
\State $(V^{(i)}_-,w^{(i)}_-) = \text{Continuous-evolution}(V^{(i)},w^{(i)},t,t_k,y,C^{(i)})$
\State $(V^{(i)},w^{(i)}) = \text{Jump}(V^{(i)}_-,w^{(i)}_-,y,y_k,C^{(i)})$
\EndFor
\State Set $U^{(i)} = (V^{(i)},C^{(i)})$ for $i=1$ to $N_s$
\If {resampling}  $(U, w)=\text{Offsprings}(U, w)$
\Else { rescale } $w$ so that $\sum_{i=1}^{N_s}w^{(i)}=N_s$
\EndIf
\State Set $(V^{(i)},C^{(i)})=U^{(i)}$ for $i=1$ to $N_s$
\State Set $t=t_k$, $y=y_k$
\EndFor
\For { $i=1$ to $N_s$}
\State $(V^{(i)},w^{(i)}) =
\text{Continuous-evolution}(V^{(i)},w^{(i)},t_{N_k},T,y_{N_k},C^{(i)})$
\EndFor
\end{algorithmic}
\end{algorithm}

\subsection{Estimating past state}
Besides the real-time update of the conditional probability mass function $\pi(t,x) = P\{X(t)=x \, | \, Y(s)=y(s), \, 0 \leq s \leq t \}$, we would like to look back and consider estimation of the state at time $t$ with the observation made until a later time $T$. A special case is when the initial state of the system is not known precisely, but rather a prior distribution is the best of our knowledge. 

To tackle the problem of finding 
\[
P\{X(t_0)=x \, | \, Y(s)=y(s), \, 0 \leq s \leq T \}
\]
with $0 \leq t_0 \leq T$, we expand the state of the underlying process as $(X,\tilde{X},Y)$ where $\tilde{X}$ is simply a copy of $X$ until time $t_0$ and there after $\tilde{X}(t)$ 
remains fixed at the value $\tilde{X}(t_0)$. 

The process $(X,\tilde{X},Y)$ with state space $\posint^{n_1} \times \posint^{n_1} \times \posint^{n_2}$ will be a piecewise time homogeneous Markov process in the sense that on the interval $[0,t_0]$ it will evolve 
according to a certain reaction network and during $(t_0,T]$ 
it will evolve according to a different reaction network. During $[0,t_0]$ whenever $X$ jumps by $\nu_j'$, $\tilde{X}$ also jumps by $\tilde{\nu}_j'=\nu_j'$. However, during $(t_0,T]$ the stoichiometric vectors $\tilde{\nu}_j'$  corresponding to $\tilde{X}$ are all zero indicating no change in state. Based on this, the original filtering equations \eqref{eq-pi-deriv}, \eqref{eq-pi-jump}, \eqref{eq-rho-deriv} and \eqref{eq-rho-jump} are still valid for this extended process. 
We are simply estimating 
\[
P\{\tilde{X}(T)=x \, | \, Y(s)=y(s), \; 0 \leq s \leq T\}
\]
which equals
\[
P\{X(t_0)=x \, | \, Y(s)=y(s), \, 0 \leq s \leq T \}
\]
since $\tilde{X}(T)=X(t_0)$.
Algorithm \ref{alg-interm-state} describes this. 

\begin{algorithm}[th]  
\caption{Past state estimation}
\label{alg-interm-state} 
\begin{algorithmic}[1]
\State \textbf{Input:} Jump times of $Y$ $(t_1,\dots,t_{N_k})$ and observed $Y$ at jump times $(y_1,\dots,y_{N_k})$. Initial distribution $\mu_0$ for $X(0)=x_0$, parameter value $c$, intermediate time $t_0$, final time $T$, filter sample size $N_s$.
\State Set $k=1$, $t=0$, $y=y_0$ 
\For { $k = 1$ to $k_1$}
\For { $i=1$ to $N_s$ }
\State $(V^{(i)}_-,w^{(i)}_-) = \text{Continuous-evolution}(V^{(i)},w^{(i)},t,t_k,y,c)$
\State $(V^{(i)},w^{(i)}) = \text{Jump}(V^{(i)}_-,w^{(i)}_-,y,y_k,c)$
\EndFor
\If {resampling}  $(V, w)=\text{Offsprings}(V, w)$
\Else { rescale } $w$ so that $\sum_{i=1}^{N_s}w^{(i)}=N_s$
\EndIf
\State Set $t=t_k$, $y=y_k$
\EndFor
\State $(V^{(i)},w^{(i)}) =
\text{Continuous-evolution}(V^{(i)},w^{(i)},t,t_0,y_{N_k},c)$ for $i = 1, \cdots, N_s$
\State Set $\tilde{V}^{(i)}=V^{(i)}$ for $i = 1, \cdots, N_s$
\For { $k = k_1 + 1$ to $N_k$}
\For { $i=1$ to $N_s$ }
\State $(V^{(i)}_-,w^{(i)}_-) = \text{Continuous-evolution}(V^{(i)},w^{(i)},t,t_k,y,c)$
\State $(V^{(i)},w^{(i)}) = \text{Jump}(V^{(i)}_-,w^{(i)}_-,y,y_k,c)$
\EndFor
\State Set $U^{(i)} = (V^{(i)}, \tilde{V}^{(i)})$ for $i = 1, \dots, N_s$
\If {resampling}  $(U, w)=\text{Offsprings}(U, w)$
\Else { rescale } $w$ so that $\sum_{i=1}^{N_s}w^{(i)}=N_s$
\EndIf
\State Set $(V^{(i)},\tilde{V}^{(i)})=U^{(i)}$ for $i = 1, \dots, N_s$
\State Set $t=t_k$, $y=y_k$
\EndFor
\State $(V^{(i)},w^{(i)}) =
\text{Continuous-evolution}(V^{(i)},w^{(i)},t,T,y_{N_k},c)$ for $i = 1, \cdots, N_s$
\end{algorithmic}
\end{algorithm}

\section{Estimation error}\label{sec-est-error}
In this section, we discuss two notions of error in state or parameter estimation. The first is a measure of the error in the estimated conditional distribution and the true conditional distribution while the second is the $L^2$ error in estimating the state or a parameter.

We first note that the conditional distribution $\pi(t,x)$ is 
a function of the observed trajectory of $Y$ from $0$ to $t$. 
For clarity, we write $\pi(t,x,Y_{[0,t]})$. Likewise, we write $\hat{\pi}(t,x,Y_{[0,t]})$ for the estimator.   

In order to measure the error between the estimated 
conditional distribution and the true conditional distribution we introduce the {\em mean total variation error} (MTVE):  
\begin{equation}\label{eq-MTVE}
    \text{MTVE} = \bE \left(\sum_{x \in \posint^{n_1}} | \hat{\pi}(t, x, Y_{[0,t]}) - \pi(t, x, Y_{[0,t]})| \right).
\end{equation}
We also introduce the {\em conditional mean total variation error} (CMTVE):  
\begin{equation}\label{eq-CMTVE}
    \text{CMTVE} = \bE \left(\sum_{x \in \posint^{n_1}} | \hat{\pi}(t, x, Y_{[0,t]}) - \pi(t, x, Y_{[0,t]})| \,\, \Big{|} \, \sY_t \right),
\end{equation}
where $\sY_t$ is the $\sigma$-algebra generated by the observed process $Y$ up to time $t$. The reader unfamiliar with 
$\sigma$-algebras may regard this as condition on the trajectory of $Y$ up to time $t$. Thus, CMTVE is a function of the observed trajectory. 
If we are estimating conditional density of a parameter as opposed to state, the summations in the above definitions need to be replaced by an integrals. 

On the other hand, instead of attempting to describe the entire conditional distribution for the (unobserved) state or a parameter, one is frequently interested in obtaining a point estimate. The ideal point estimate of the unobserved state $X(t)$ is the conditional expectation 
$\bE[X(t) \, | \, \sY_t]$. In practice one estimates the latter via 
the estimator $\sE(t)$ defined by 
\begin{equation}\label{eq-sE}
\sE(t) = \frac{\sum_{i=1}^{N_s} V^{(i)}(t) \, w^{(i)}(t)}{\sum_{i=1}^{N_s} w^{(i)}(t)}. 
\end{equation}
Thus the estimation error is given by $e(t)=\sE(t)-X(t)$. 
The quantities of interest are, the conditional bias $\bE[e(t) \, | \, \sY_t]$, the bias $\bE[e(t)]$, the conditional $L^2$ error 
$(\bE[ e^2(t) \, | \, \sY_t])^{1/2}$ and the $L^2$ error $(\bE[e^2(t)])^{1/2}$.

It is instructive to split $e(t)$ as 
\[
e(t) = \left(\sE(t) - \bE[X(t) \, | \, \sY_t] \right) +  \left( \bE[X(t) \, | \, \sY_t] - X(t) \right).
\]
We make the {\em important observation that $\sE(t)$ and $X(t)$ are independent when conditioned on $\sY_t$.} Hence
\[
\begin{aligned}
&\bE\Big{\{}\Big{(}\sE(t)-\bE[X(t)|\sY_t]\Big{)} \, \Big{(}\bE[X(t)|\sY_t] - X(t)\Big{)} \; \Big{|} \; \sY_t \Big{\}} \\
&= \bE\Big{\{} \Big{(}\sE(t)-\bE[X(t)|\sY_t]\Big{)} \, \Big{|} \, \sY_t \Big{\}}\, \bE\Big{\{} \Big{(}\bE[X(t) | \sY_t] - X(t)\Big{)} \, \Big{|} \, \sY_t \Big{\}}
= 0.
\end{aligned}
\]
Hence we may expand
\begin{equation}\label{eq-cond-L2}
\begin{aligned}
\bE[e^2(t) \, | \, \sY_t] &= \bE\left[ \left(\sE(t) - \bE[X(t)|\sY_t]\right)^2 \, | \sY_t\right] + \bE\left[ \left(\bE[X(t)|\sY_t]-X(t)\right)^2 \, |\sY_t\right]\\
&= \bE\left[ \left(\sE(t) - \bE[X(t)|\sY_t]\right)^2 \, | \, \sY_t\right] + \text{Var}\left[X(t) \, | \, \sY_t\right]
\end{aligned}
\end{equation}
We observe that the second term depends only on the filtering problem and not on the filtering algorithm, while the first term depends on the filtering algorithm. We expect as the filter sample size $N_s$ approaches infinity the first term to approach zero in some sense. Taking expectation on \eqref{eq-cond-L2} we obtain that
\begin{equation}
\bE[e^2(t)] = \bE\left[ \left(\sE(t) - \bE[X(t)|\sY_t]\right)^2\right] + \bE\left[\text{Var}[X(t)\, | \, \sY_t]\right],
\end{equation}
where as observed earlier, the second term depends only on the filtering problem and provides a lower bound on the $L^2$ error.

In situations where observations are obtained via simulations as is the case in this paper, we do know the true value $X(t)$ of the unobserved state and thus $\bE[e^2(t)]$ may be estimated as follows. We simulate the original system via the Gillespie algorithm $N_r$ independent times to obtain $(X^{(1)},Y^{(1)}),\dots,(X^{(N_r)},Y^{(N_r)})$. For each such simulation $j$, we run the overall filter algorithm (with some sample size $N_s$) for the observed trajectory $Y^{(j)}$ to record $\sE_j(t)$, the filter estimate of the conditional expectation. Then we may estimate $\bE[e^2(t)]$ as
\[
\bE[e^2(t)] \approx \frac{1}{N_r} \sum_{j=1}^{N_r} |\sE_j(t) - X^{(j)}(t)|^2.
\]

On the other hand, for any given observation trajectory, estimation of the conditional error $\bE[e^2(t) \, | \sY_t]$ is harder even in situations as in this paper where observations are generated via Monte Carlo simulations. 
If the filter sample size $N_s$ is very large, we may be justified in approximating $\bE[e^2(t) \, | \sY_t]$ 
by $\text{Var}[X(t) \, | \sY_t]$ according to 
\eqref{eq-cond-L2}. This suggests the following estimation of $\bE[e^2(t) \, | \, \sY_t]$. Generate an observation $Y$ of one trajectory and apply the filter with sample size $N_s$ once to obtain an estimate $\mathcal{V}(t)$
of $\text{Var}[X(t) \, | \sY_t]$. Here 
\begin{equation}\label{eq-filt-cond-var}
   \mathcal{V}(t) = \frac{\sum_{i=1}^{N_s} (V^{(i)}(t))^2 \, w^{(i)}(t)}{\sum_{i=1}^{N_s} w^{(i)}(t)} -  \left(\frac{\sum_{i=1}^{N_s} V^{(i)}(t) \, w^{(i)}(t)}{\sum_{i=1}^{N_s} w^{(i)}(t)}\right)^2.  
\end{equation}

\section{Numerical Examples}\label{sec-examples} 
In this section we illustrate the filtering algorithms via examples. In all examples, observations were made by simulation of the underlying reaction network using the Gillespie algorithm \cite{Gillespie77} to obtain one or more independent samples of $(X,Y)$, the unobserved and observed trajectories. Then the filtering algorithms were applied to the observed trajectories.  

\subsection{A linear propensity example}\label{ex-linear-prop}
We consider the simple reaction network 
\begin{equation}
\begin{aligned}
    S &\stackrel{c_1}{\longrightarrow} S + A,\\
    \varnothing &\stackrel{c_2}{\longrightarrow} S,\\
    S &\stackrel{c_3}{\longrightarrow}  \varnothing,\\
\end{aligned}
\end{equation}
consisting of two species and three reaction channels and assume mass action form of propensities.
Thus the propensities are given by $a_1(z) = c_1 z_2$, $a_2(z) = c_2$, $a_3(z) = c_3 z_2$. 

First, consider the case where the copy number of species $S$ is observed exactly while the species $A$ is unobserved. Thus, $Z(t) = (X(t), Y(t)) = (\#A(t), \#S(t))$. In this situation, since the second two reaction channels alone determine the copy number of $S$, the conditional probability density function $\pi(t, x)$ could be computed exactly as
\[
\pi(t,x) = \frac{\lambda^x e^{-\lambda}}{x!}
\]
where $\lambda = \int_0^t c_1 y(t) dt$. 
With initial state $\#A = 0$ and $\#S =5$, parameter values $c= (1,5,1)$, and filter sample size $N_s=10,000$,
we ran Algorithm \ref{alg-overall}, and Figure \ref{fig-cpmf-ex1} shows the comparison between the conditional distribution computed by the filter with the exact conditional distribution.
\begin{figure}[htbp]
    \includegraphics[width=0.8\textwidth]{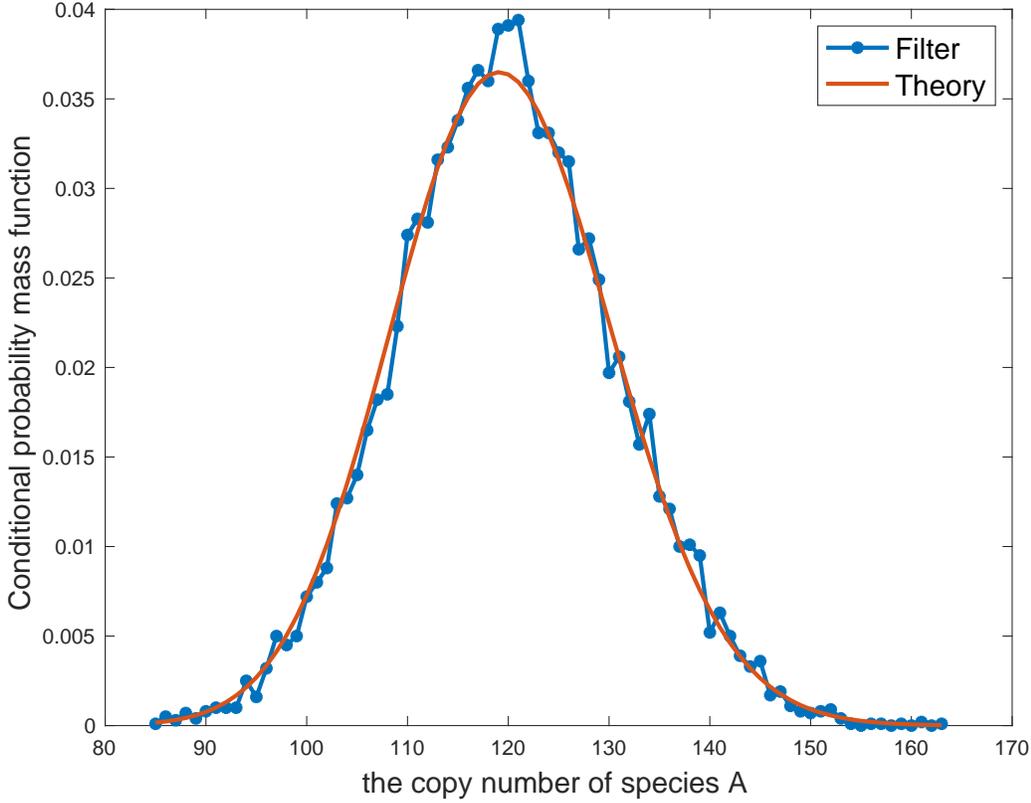}
    \caption{The conditional distribution of the number of $A$ at time $T=20$ conditioned on one observation of the entire trajectory of the copy number of $S$ on the interval $[0,T]$, in the linear propensity example. 
    The estimated conditional distribution $\hat{\pi}(T,x)$ and the exact theoretical conditional distribution $\pi(T,x)$ are shown. 
    }\label{fig-cpmf-ex1}
\end{figure}
In order to estimate the CMTVE defined by \eqref{eq-CMTVE}, with same initial state, parameter values, final time $T=20$, filter sample size $N_s=10,000$, and the same observation trajectory, we ran 100 simulations of algorithm \ref{alg-overall}. From 
this we estimated CMTVE for the algorithm to be $0.0475$ with a $95\%$confidence interval of $[0.0461, 0.0488]$.

We also tested the filter's performance on point estimation. To compare estimated state $\sE(t)$ with the actual state $X(t)$, as described in Section \ref{sec-est-error}, we generated $N_r=500$ independent realizations of the system $(X,Y)$, and applied the filtering algorithm with filter sample size $N_s=1000$ to get an estimate $\sE(t)$ for each of the realizations. Figure \ref{fig-linprop-a-scatter} shows a scatter plot of the values of $\sE(T)$ against $X(T)$. We note that $\sE(T)$ is a biased estimator whose bias 
is expected to approach zero as $N_s$ tends to infinity. The fact that the regression line does not have slope 1 is due to this as well as due to finite sample size of $N_r$. 

\begin{figure}[htbp]
    \includegraphics[width=0.7\textwidth]{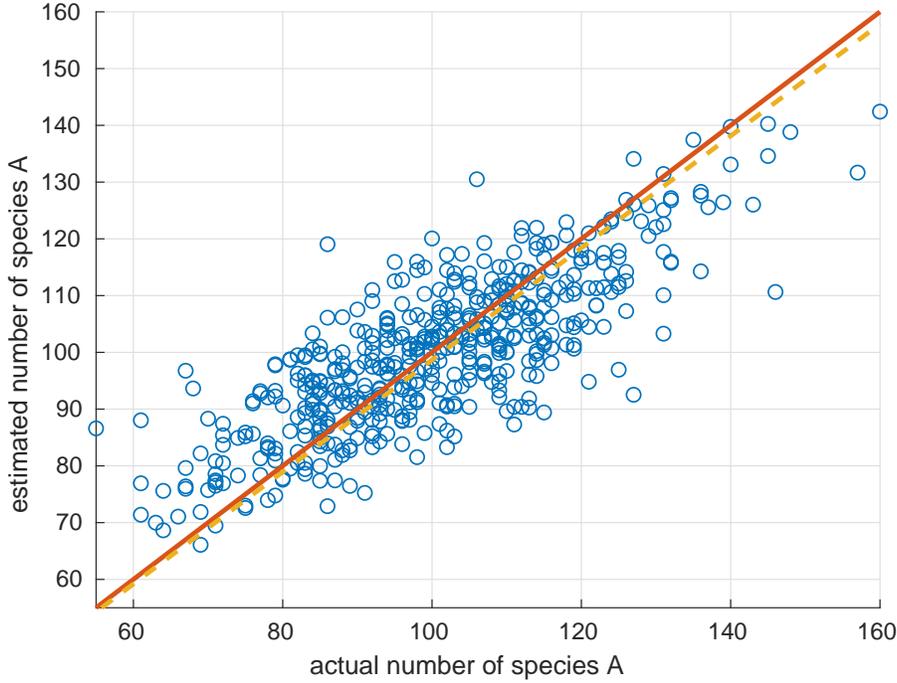}
    \caption{Scatter plot of the point estimates $\sE(T)$ of the copy number of species $A$ at $T=20$ (based on the observation of species $S$ on the interval $[0,T]$) against the actual copy numbers $X(T)$ in the linear propensity example. In each trial, a realization of the system $(X,Y)$ was generated, and we computed the point estimates $\sE(T)$ of the copy number of species $A$ at time $T$ based on the trajectory of species $S$. $500$ independent trials were performed.  
    The dashed yellow line is the regression line and solid red line is the line with slope 1.  
    The bias is $-0.2810$ with a $95\%$ confidence interval of $[-1.20427, 0.642256]$ the estimated $L^2$ error is $10.3159$ with a $95\%$ confidence interval of $[9.58143, 11.0014]$.}\label{fig-linprop-a-scatter}
\end{figure}

Next, with the same system, we considered observing species $A$ rather than species $S$. So $Z(t) = (X(t), Y(t)) = (\#S(t), \#A(t))$.
We kept the initial state and parameter values of $c_1$ and $c_3$ to be the same as before 
and performed a Bayesian estimation of parameter $c_2$. Thus we considered $c_2$ as a random variable $C_2$ with a uniform prior distribution on $[4,6]$. We randomly generated a sample of $N_r=1000$ values of $C_2$ following uniform distribution $[4, 6]$ and generated one observation trajectory $Y$ for each value of $C_2$. Then we applied Algorithm \ref{alg-overall-bayes} to compute $\bar{C}_2(t)$, the filter estimate of $\bE(C_2 \, | \sY_t)$ (see \eqref{eq-C-bar}.
The result is shown in Figure \ref{fig-linprop-c2-scatter}.


\begin{figure}[htbp]
    \includegraphics[width=17cm]{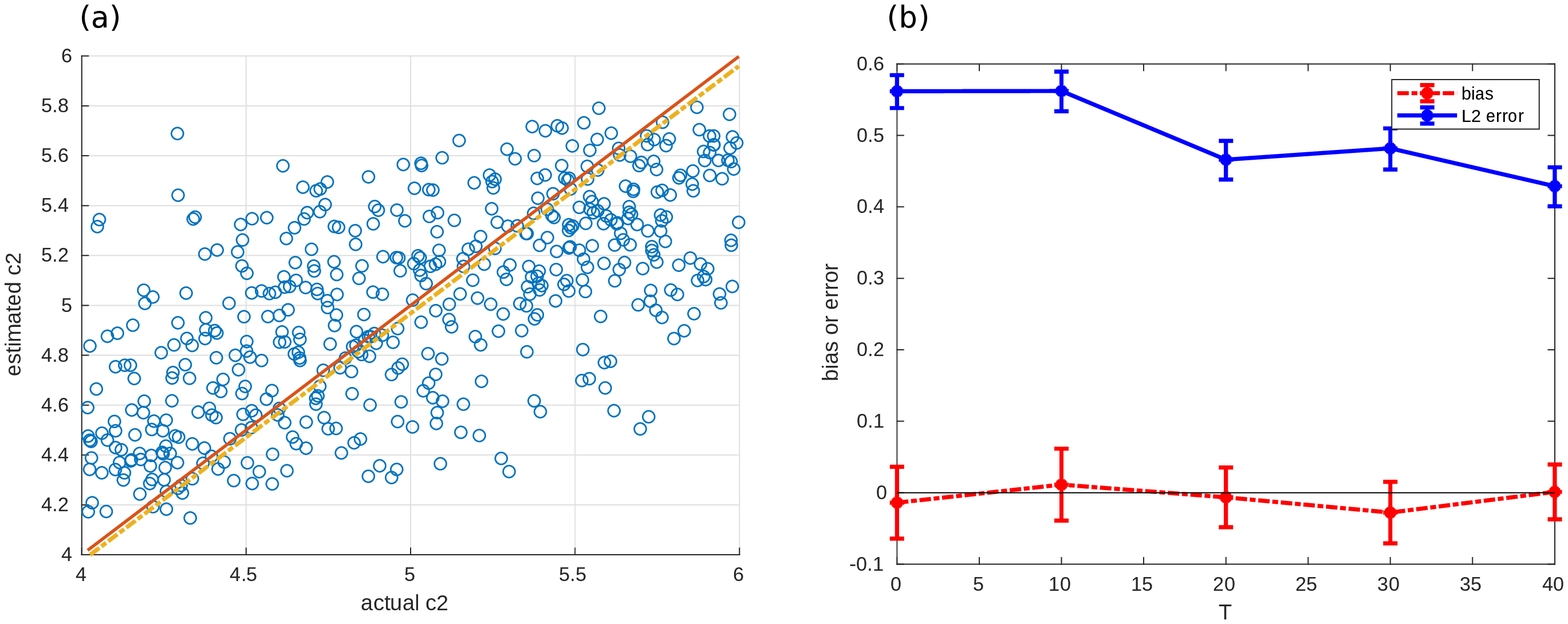}
    \caption{\label{fig-linprop-c2-scatter} (a) A scatter plot of $\bar{C}_2(40)$, the estimated value of $C_2$ (based on the observation of the copy number trajectory of species $A$ on the interval $[0,40]$), against the actual value of $C_2$ in the linear propensity example. In each trial, we randomly generated a value of $C_2$ following a uniform distribution on $[4, 6]$ and simulated the full system to generate an observation trajectory of species $A$. Then we applied Algorithm \ref{alg-overall-bayes} to compute $\bar{C}_2(40)$, where a uniform (prior) distribution on the interval $[4,6]$ was endowed. Filter sample size was $N_s = 1000$ and trial size $N_r=500$. (b) The bias $\bE[\bar{C}_2(T)-C_2]$  as well as the $L^2$ error $\bE[(\bar{C}_2(T)-C_2)^2]^{1/2}$ (along with $95\%$ confidence intervals) are plotted against $T$. Decreasing $L^2$ error with $T$ suggests that longer observations lead to better estimations.}
\end{figure}

\subsection{Genetic Circuit Example}\label{ex-circuit}
We consider a genetic transcription example \cite{rathinam2007reversible} where a protein $A$ encoded by gene $D_A$ could bind with its own gene promoter to enhance or suppress its transcription:
\begin{equation}
\begin{aligned}
D_A + A &\stackrel{c_1}{\longrightarrow} D_A',\\
D_A' &\stackrel{c_2}{\longrightarrow} D_A + A,\\
D_A &\stackrel{c_3}{\longrightarrow} D_A + A, \\
D_A' &\stackrel{c_4}{\longrightarrow} D_A' + A,\\
A &\stackrel{c_5}{\longrightarrow} \varnothing.
\end{aligned}
\end{equation}
Let $Z(t) = (\#D_A(t), \#D_A'(t),\#A(t))$ and suppose the propensity function has the mass action form $a_1(z) = c_1 z_1 z_3, a_2(z) = c_2 z_2, a_3(z) = c_3 z_1, a_4(z) = c_4 z_2, a_5(z) = c_5 z_3$. Suppose we only observe the molecule counts of protein $A$ and would like to estimate the copy number of the naked form of gene  promoter $D_A$ and the bounded form of the gene promoter $D_A'$. We ran all numerical experiments with initial condition $Z(0) = (3, 0, 15)$ and nominal parameter values $ c_1 = 0.3, c_2 = 3, c_3 = 0.5, c_4 = 0.2, c_5 =0.06$ and filter sample size $N_s=10,000$. 

The estimation of the copy number of $D_A$ at time $t$ based 
on observation of $A$ over $[0,t]$ is shown in  Figure \ref{fig-circuit-state}.


\begin{figure}[htbp]
    \includegraphics[width=0.45\textwidth]{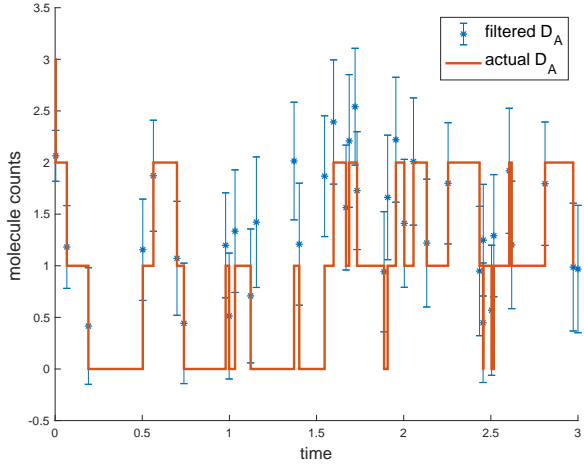}
    \caption{\label{fig-circuit-state} 
    Estimation of copy number of $D_A$ at time $t$ based on the observation of species $A$ over $[0,t]$ (genetic circuit example). Estimates are shown at jump times of the observed species $A$ along with their $68 \%$ confidence intervals.
    Filter sample size used was $N_s = 10,000$. }
\end{figure}
We estimated one parameter at a time while fixing the other parameters at their nominal values mentioned before. We chose a uniform distribution as the prior and the posterior conditional distributions corresponding to the same single observed trajectory at several snapshots in time are shown in Figure \ref{fig-circuit-para}.
\begin{figure}[htbp]
    \includegraphics[width=17cm]{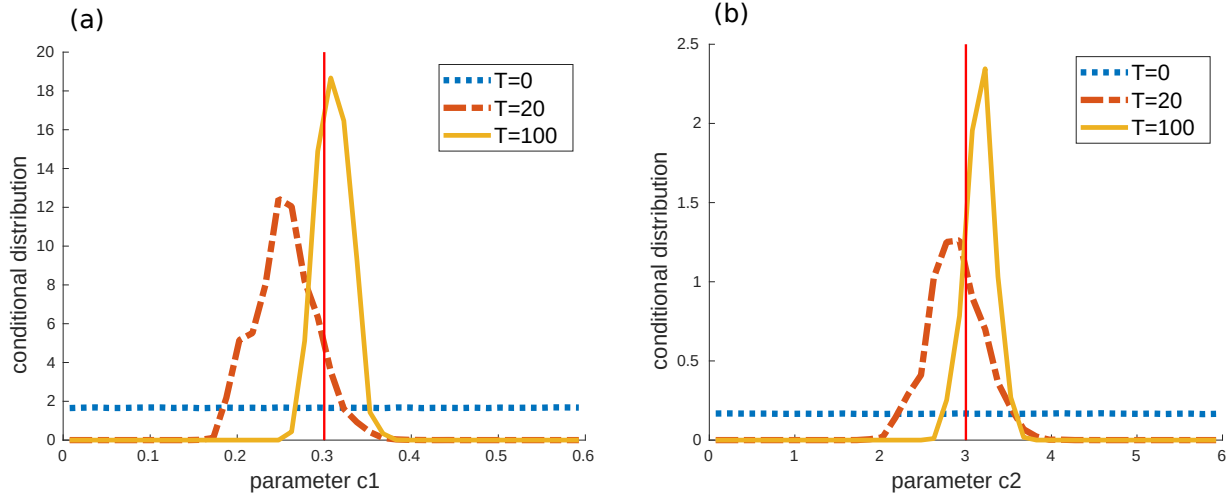}
    \caption{\label{fig-circuit-para} The posterior conditional probability density function of parameters $C_1$  and $C_2$  based on one observed trajectory of $A$ on interval $[0,T]$ in the genetic circuit example for $T=0, 20$ and $100$. The filter sample size was $N_s = 1,000,000$. $T=0$ is simply the uniform priors on $C_1$ and $C_2$. The vertical lines show the true parameter values.}
\end{figure}

\subsection{Genetic toggle switch}\label{ex-toggle}
Consider the system of genetic toggle switch\cite{gardner2000construction}
\begin{equation}
\begin{aligned}
\varnothing &\stackrel{a_1}{\longrightarrow} S_1,\\
S_1 &\stackrel{a_2}{\longrightarrow} \varnothing,\\
\varnothing &\stackrel{a_3}{\longrightarrow} S_2, \\
S_2 &\stackrel{a_4}{\longrightarrow} \varnothing.
\end{aligned}
\end{equation}

Let $Z(t) = (X(t),Y(t))=(\#S_1(t), \#S_2(t))$ and suppose the propensity functions have the form $a_1(z) = \frac{\alpha_1}{1+y^\beta}$, $a_2(z) = x$, $a_3(z) = \frac{\alpha_2}{1+x^\gamma}$, $a_4(z) = y$. Suppose we only observe species $S_2$ and we estimate the molecule counts of species $S_1$. We chose the initial condition $Z(0) = (0,0)$ and the nominal parameter values $\alpha_1 = 50, \alpha_2 = 16, \beta = 2.5 , \gamma = 1$. The estimation of the molecular counts of species $S_1$ right after each jump of $\#S_2(t)$ for a particular observed trajectory is shown in Figure \ref{fig-toggle-switch} where a filter sample size of $N_s=10,000$ was used.


As in the genetic circuit example, we estimated one parameter at a time while fixing the other parameter at its nominal value mentioned before. We chose a uniform distribution as the prior and the posterior conditional distributions computed based on the same single observed trajectory at several snapshots in time are shown in Figure \ref{fig-toggle-para}.

Figure \ref{fig-toggle-para} suggests that the inference of $\alpha_2$ is much better than the inference of $\alpha_1$. To understand why, we note that in the toggle switch example, the trajectories of $\#S_1$ and $\#S_2$ 
exhibit a switching pattern that switches between two modes. In one mode, $\#S_1$ fluctuates around $\alpha_1$ while $\#S_2$ is nearly zero, and in the other, $\#S_2$ fluctuates around $\alpha_2$ while 
$\#S_1$ is nearly zero. The influence of $\alpha_1$ on the behavior of $S_2$ only enters indirectly, and we conjecture by affecting the switching frequency. 
If this were to be the case, one may need a long observation trajectory, perhaps beyond $T=10,000$ to make a good estimate of $\alpha_1$. Longer time duration $T$ may require larger sample size $N_s$ 
making the computations more tedious. 

In order to test our hypothesis regarding the switching frequency, after some trial and error, we chose the parameter values $\alpha_1=20$ and $\alpha=9$ and kept the other two parameters the same. Figure \ref{fig-traj-toggle} shows a comparison of trajectories of the system with the two different sets of parameters.
We repeated the parameter inference experiment with this new choice of $\alpha_1$ and $\alpha_2$ and the results are shown in Figure \ref{fig-toggle-para-2} and show that $\alpha_1$ is predicted better.    


\begin{figure}[htbp]
    \includegraphics[width=0.7\textwidth]{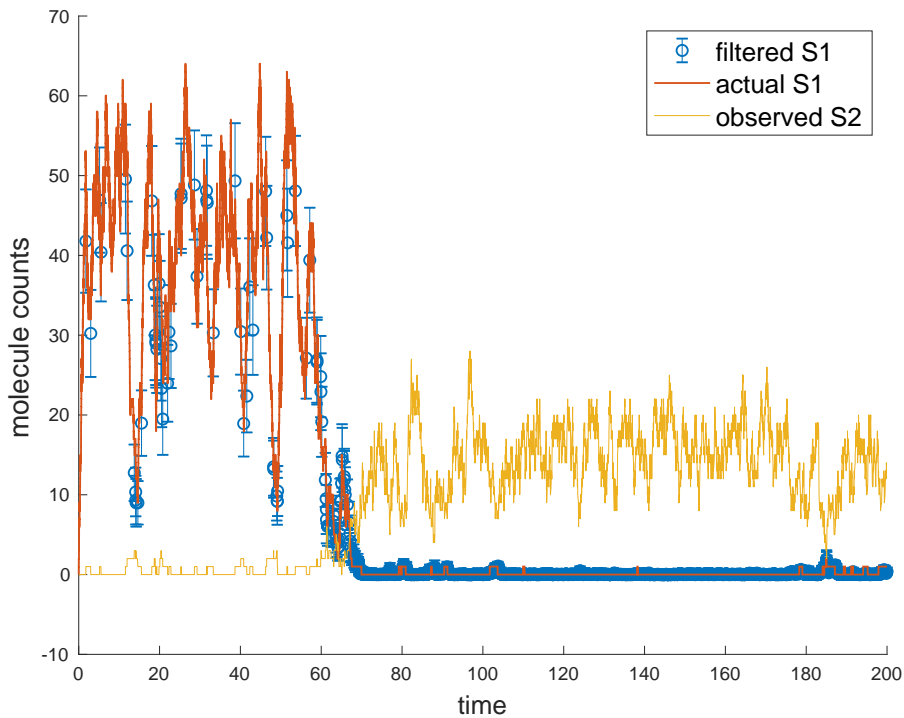}
    \caption{\label{fig-toggle-switch} Estimation of copy number of $S_1$ at time $t$ in the genetic toggle switch example based on one observed trajectory of $S_2$ over $[0,t]$. Estimates are shown at jump times of the observed species $S_2$ along with their $68 \%$ confidence intervals. Filter sample size used was $N_s = 10,000$.}
\end{figure}



\begin{figure}[htpb]
    \includegraphics[width=16cm]{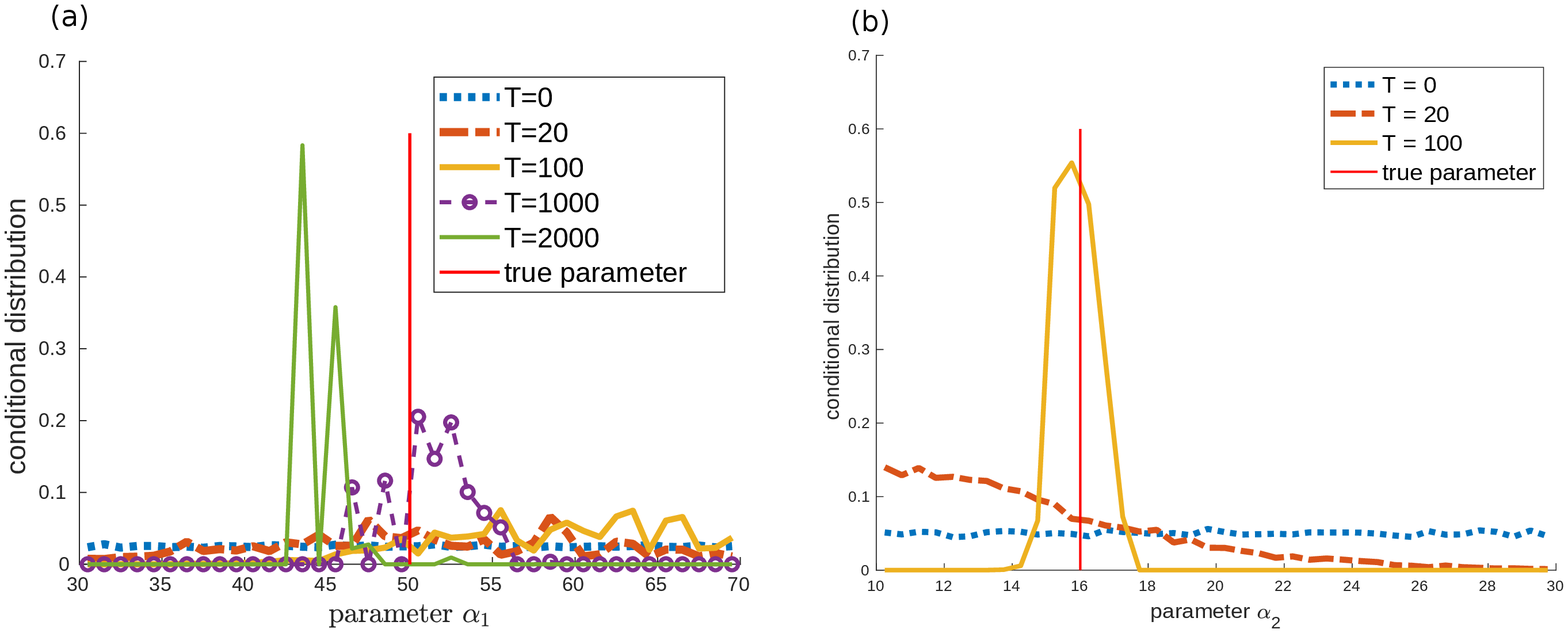}
    \caption{\label{fig-toggle-para} The posterior conditional probability density function of parameters $\alpha_1$ (left) and $\alpha_2$ (right) in the genetic toggle example based on the observation of the trajectory of $\#S_2$ over $[0,T]$. The nominal parameter values were $\alpha_1 = 50, \alpha_2 = 16, \beta = 2.5 , \gamma = 1$. Estimates are shown at times $T=0, 20, 100, 1000$ and $2000$ on the left and $T=0, 20$ and $100$ on the right. The filter sample size was $N_s = 10,000$. Note that $T=0$ corresponds to the uniform priors. Two different observation trajectories, one for the estimation of $\alpha_1$ and the other for $\alpha_2$ were used. }
\end{figure}

\begin{figure}[htbp]
    \includegraphics[width=17cm]{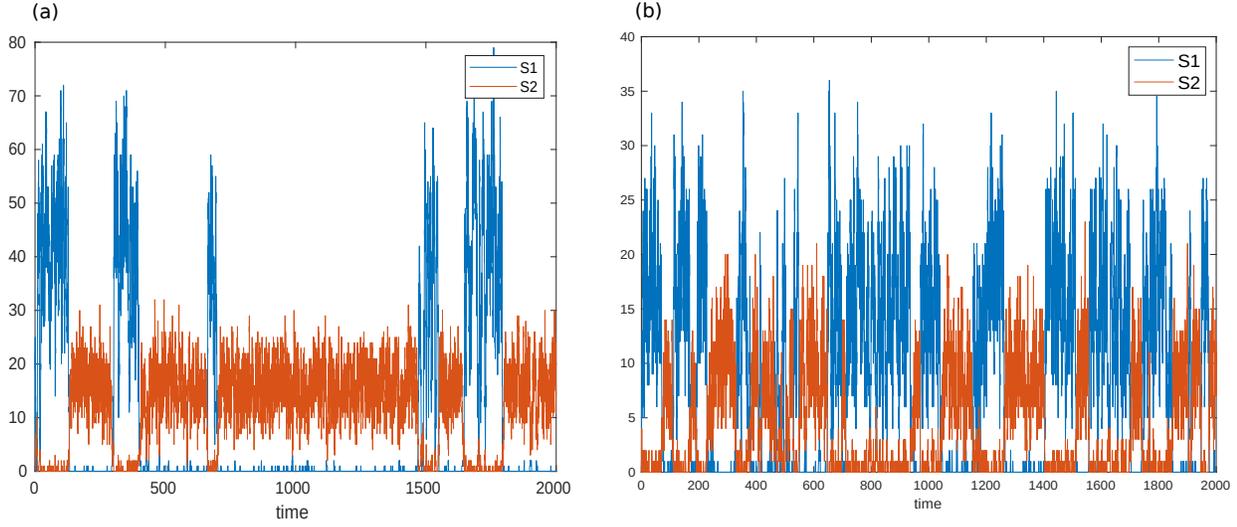}
    \caption{Trajectories of the genetic toggle switch system under different parameters. The figure on the left has less frequent switching, which corresponds to $\alpha_1 = 50, \alpha_2 = 16, \beta = 2.5 , \gamma = 1$, the figure on the right corresponds to  $\alpha_1 = 20, \alpha_2 = 9, \beta = 2.5 , \gamma = 1$.}
    \label{fig-traj-toggle}
\end{figure}

\begin{figure}[htbp]
    \includegraphics[width=16cm]{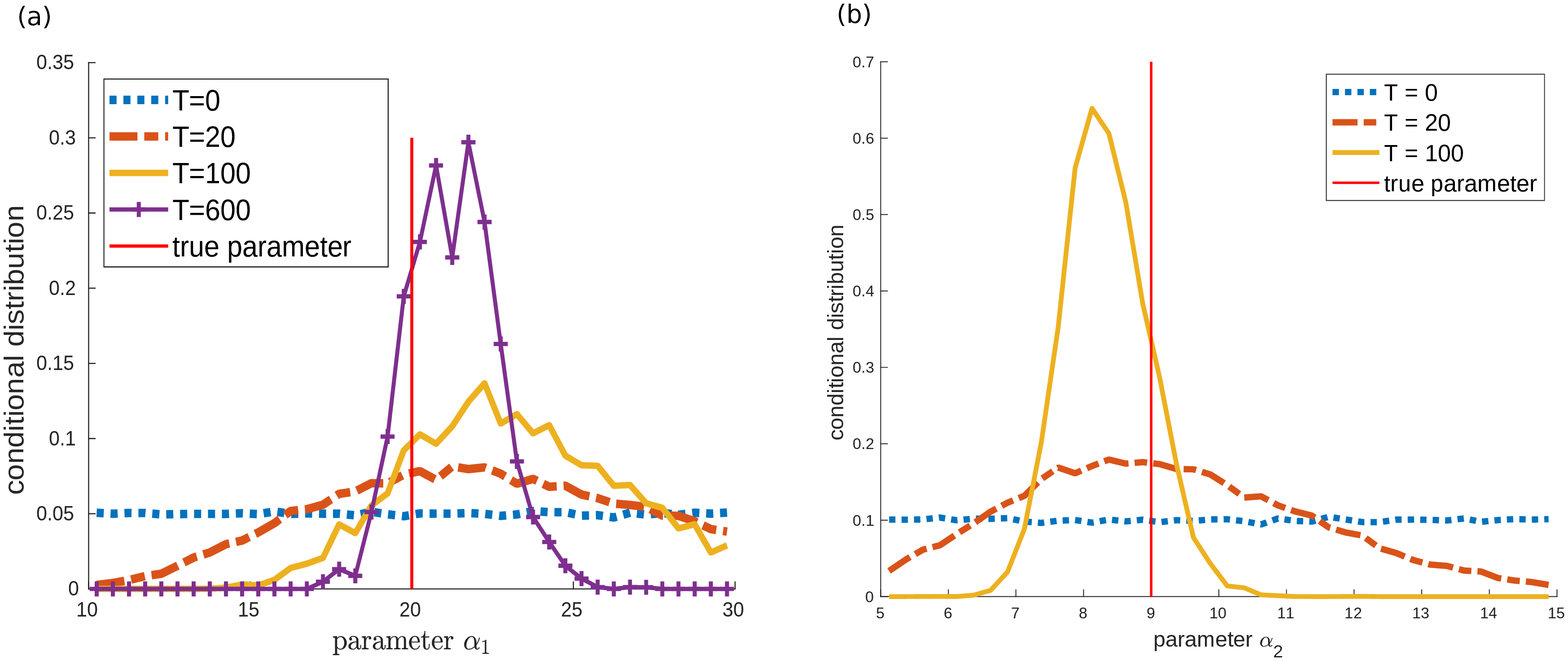}
    \caption{\label{fig-toggle-para-2} The posterior conditional probability density function of parameters $\alpha_1$ (left) and $\alpha_2$ (right) in the genetic toggle example based on the observation of the trajectory of $\#S_2$ over $[0,T]$. The nominal parameter values were $\alpha_1 = 20, \alpha_2 = 9, \beta = 2.5 , \gamma = 1$. Estimates are shown at times $T=0,20,100$ and $600$ on the left and $T=0, 20$ and $100$ on the right. The filter sample size was $N_s = 100,0000$. Note that $T=0$ corresponds to the uniform priors.}
\end{figure}

\subsection{SEIR model}\label{ex-seir}
Consider the SEIR model for the spread of infectious diseases
\begin{equation}
\begin{aligned}
S+I &\stackrel{\beta}{\longrightarrow} E + I,\\
E &\stackrel{\kappa}{\longrightarrow} I,\\
I &\stackrel{\gamma}{\longrightarrow} R.\\
\end{aligned}
\end{equation}
Let $Z(t) = (\#S(t), \#E(t), \#R(t), \#I(t))$ 
where $S,E,I$ and $R$ stand for susceptible, exposed, infected and recovered respectively. We assume that we could observe the infectious population $I$ exactly. The propensity functions $a_j(\cdot)$ have the form $a_1(z) = \beta z_1 z_4/N$, $a_2(z) = \kappa z_2$, $a_3(z) = \gamma z_4$, where $N = z_1 + z_2 + z_3 + z_4$ is the total population. We ran all numerical experiments with initial condition $Z(0) = (s_0, e_0, r_0, i_0) = (500, 20, 0, 5)$ and parameter values $\beta = 0.05, \kappa = 0.2, \gamma = 0.05$. 

The estimation of the susceptible and exposed population after each change of $Z_4(t)=\#I(t)$ for a particular observed trajectory of $I$ is shown in Figure \ref{fig-seir-traj} where a filter sample size of $N_s=10,000$ was used. Figure \ref{fig-seir-s-scatter} shows scatter plots corresponding to state estimation of susceptible and exposed number of individuals. 

Keeping other parameters at their nominal values, we explored Bayesian inference of parameter $\kappa$ (the reciprocal of the incubation period) based on a uniform prior. The scatter plot of the parameter estimation of $\kappa$ based on $N_r = 1000$ trials each with a filter sample size $N_s=1000$ is in shown in Figure \ref{fig-seir-scatter}.

\begin{figure}[htbp]
    \includegraphics[width=17cm]{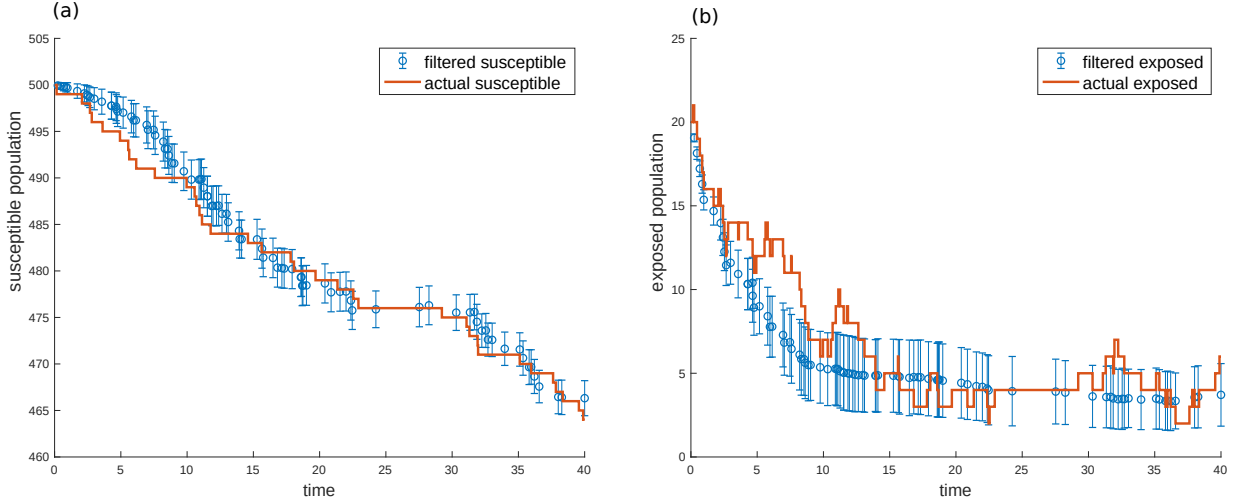}
    \caption{\label{fig-seir-traj}
    Estimation of the susceptible population (left) and the exposed population (right) at time $t$ in the SEIR example given one trajectory of the infected population over $[0,t]$. Estimates are shown at times of observed new infections or recovery along with the $68 \%$ confidence intervals. Filter sample size used was $N_s = 10,000$.}
\end{figure}
\begin{figure}[htbp]
    \includegraphics[width=17cm]{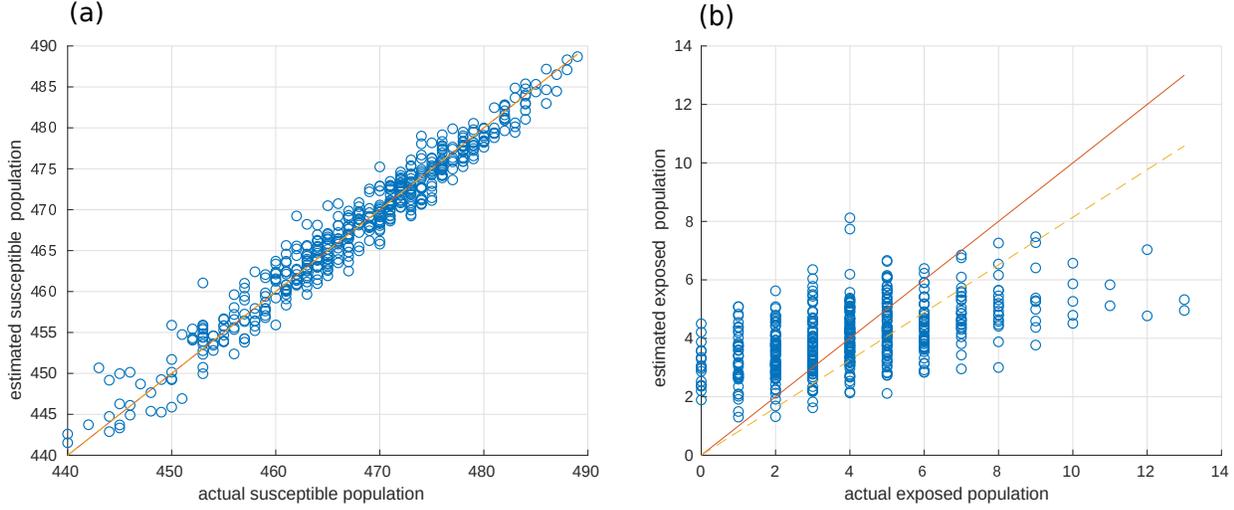}
    \caption{Estimate of the susceptible (a) and exposed (b) population in SEIR example at time $T=40$ given observations of the infected population over $[0,T]$. The filter sample size was $N_s = 1000$, and trial size $N_r = 500$. For the susceptible population, the bias is  $0.0156$ within a $95\%$ confidence interval $[-0.160955, 0.192202]$ and $L_2$ error is $1.9723$ within $[1.81684, 2.11637]$. For the exposed population, the bias is  $-0.0156$ within a $95\%$ confidence interval $[-0.192202, 0.160955]$ and $L_2$ error is $1.9723$ within $[1.81684, 2.11637]$.
    }\label{fig-seir-s-scatter}
\end{figure}
\begin{figure}[htbp]
    \includegraphics[width=17cm]{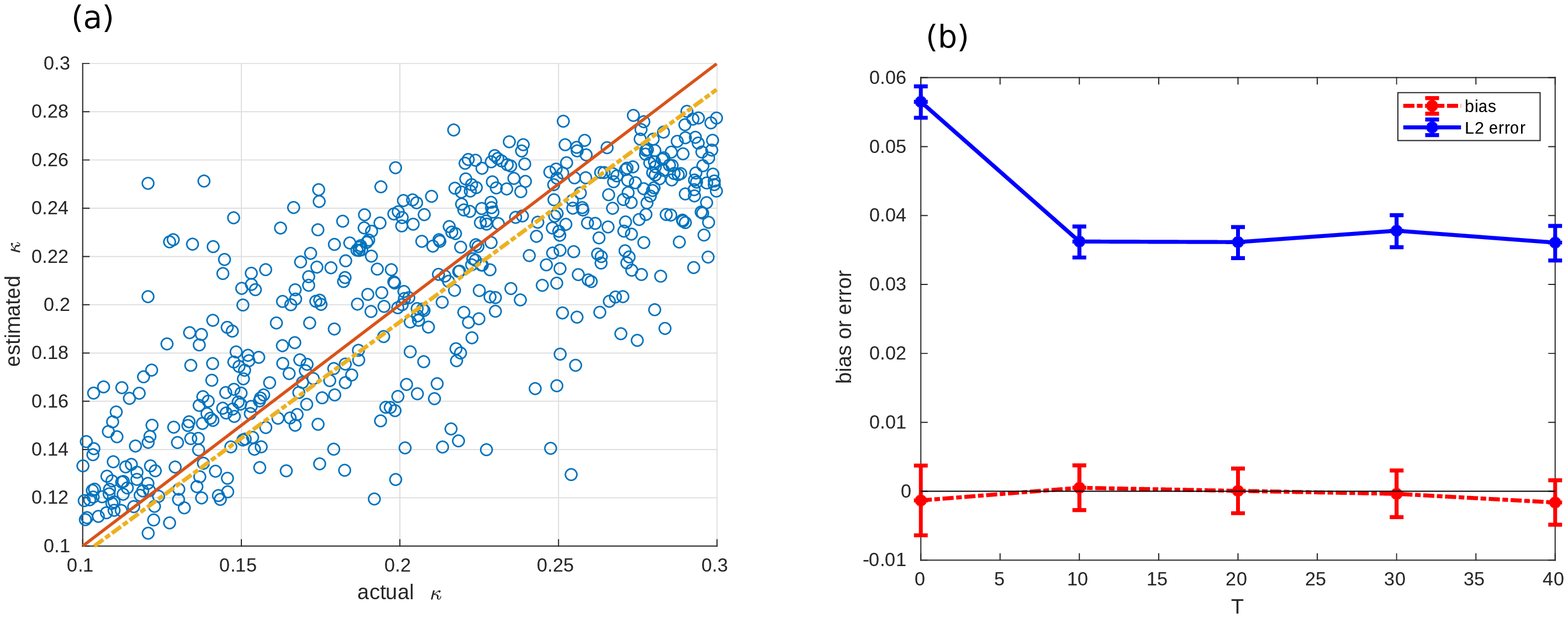}
    \caption{\label{fig-seir-scatter}
    (a) A scatter plot of $\hat{\kappa}(T)$, the estimated value of $\kappa$ against the actual value of $\kappa$. Here $T=40$. 
    Filter sample size was $N_s = 1000$ and trial size $N_r=500$. (b) The the bias $\bE[\hat{\kappa}(T)-\kappa]$ as well as the  $L^2$ error $\bE[(\hat{\kappa}(T)-\kappa)^2]^{1/2}$ (along with $95\%$ confidence intervals) are plotted against $T$.}
\end{figure}

Next we explore estimation of the state at a past time $t_0$ based on observation up to a future time $T$. That is $0 \leq t_0 \leq T$. The case of $t_0=0$ corresponds to the situation where the initial state itself is not known, but we only have a prior distribution for it. 
We assumed that the exposed population at initial time $t=0$ followed a binomial distribution, with parameters $N=520$ and $p=0.04$. Since $\#I(0)=5$ and the total population is $525$, this left us with $520$ individuals. We assumed a probability $p=0.04$ of being exposed, which lead us to choose the binomial distribution. Algorithm \ref{alg-interm-state} was used to estimate the conditional probability mass function 
\[
P\{X_2(t_0)=x \, | \, Y(s) = y(s), \, 0 \leq s \leq T\}
\]
for various $t_0$ and $T$. The results are shown in  Figure \ref{fig-seir-init}.
\begin{figure}[htbp]
    \includegraphics[width=17cm]{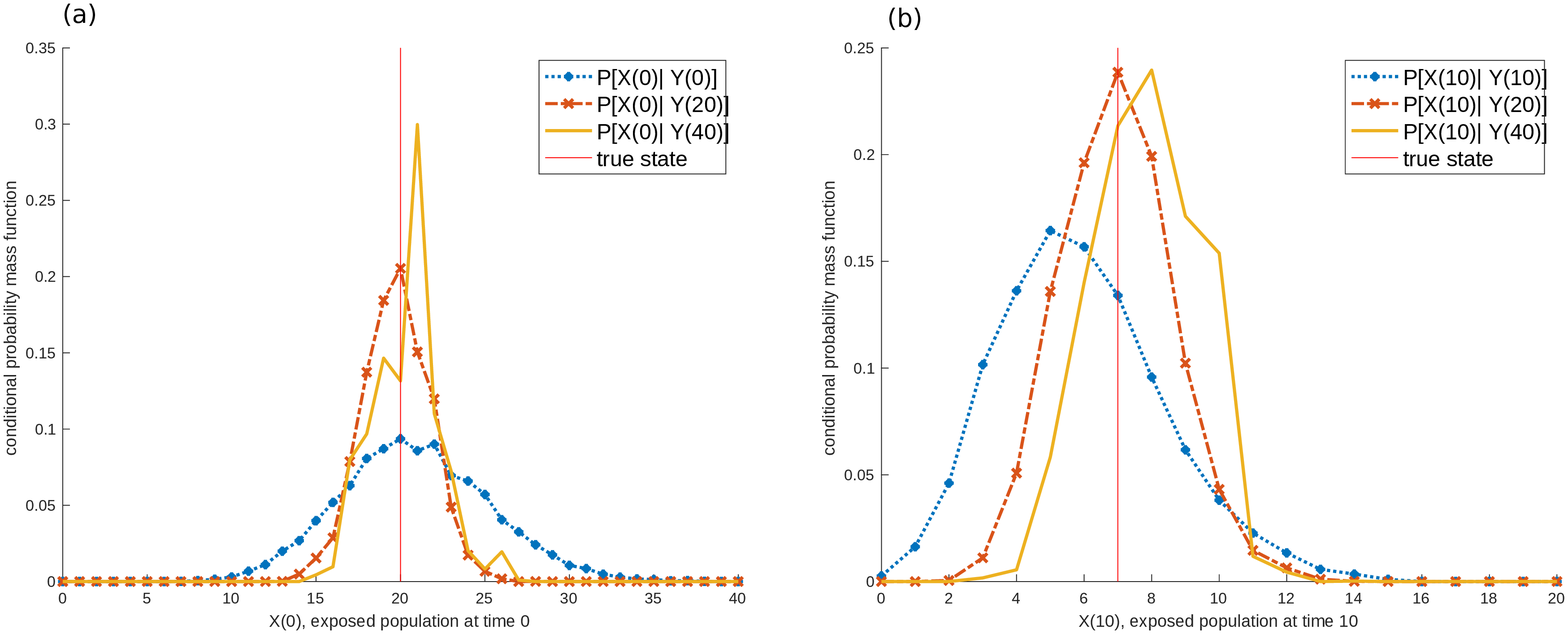}
    \caption{\label{fig-seir-init}
    (a) Conditional distribution of the exposed population at the initial time $X_2(0) | \mathcal{Y}_T$ for $T=0,20$ and $40$. (b) Conditional distribution of the exposed population at an intermediate time $X_2(10) | \mathcal{Y}_T$ for $T=10, 20$ and $40$. Filter sample size was $N_s = 10,000$. The prior of the initial distribution of the exposed population follows a binomial distribution $B(520, 0.04)$.} 
\end{figure}

\subsection{Investigation of resampling}\label{sec-resampling}

Here we investigate the strategies of resampling in our filter. In the numerical examples presented so far, we applied the offspring algorithm to resample right after each jump $t_k$. 
Here we present the results of three different resampling strategies. One was to resample at each jump time $t_k$ as we did earlier. 
As an alternative, we applied an adaptive resampling method where resampling is applied at $t_k$ if more than 10 particles had zero weights or if the ratio of the maximum weight to the minimum nonzero weight was greater than 1000. The third alternative was not to resample at all. We note that, whenever resampling was not used, the weights were normalized instead so that the average weight was 1 right after any jump time $t_k$.

Table \ref{tab:rsmpl} shows the $L^2$ error as well as the bias in estimating the conditional mean of the state 
or a parameter along with the $95\%$ confidence intervals, for the various examples considered earlier. The results are for state estimation unless a parameter is mentioned in the first column. The entrees {\tt NaN} indicate numerical issues in {\tt MATLAB} mainly due to the fact that all weights became either zero or infinity within numerical precision. 
These results show that the three approaches yield more or less the same accuracy (the estimated errors are within the confidence intervals of each other) except in the case of the genetic circuit example where serious numerical issues manifest without resampling. We suspect that in the other examples, if final time $T$ is increased, the need for resampling will become evident. 

\begin{table*}
\caption{\label{tab:rsmpl} The bias and error of estimation of the states or parameters when different resampling schemes were applied.}
\begin{tabular}{c c c c c c}
\hline
System & Method & $L^2$ error & Confidence interval & bias & Condidence interval \\
\hline
linear propensity & each jump & $9.7470$ & $[9.0677, 10.3819]$ & $-0.3346$ & $ [-1.2068, 0.5375]$\\
linear propensity & adaptive $(0\%)$ & $ 10.0216$ & $[9.3373, 10.6621]$ & $-0.0699$ & $[-0.9671, 0.8273]$\\
linear propensity & never & $ 10.0783$ & $[9.4464, 10.6728]$ & $-0.4910$ & $[-1.3922, 0.4103]$\\
genetic circuit & each jump &  $0.5638$ & $[0.5241, 0.6009]$ & $0.0100$ & $[-0.0405, 0.0604]$\\
genetic circuit & adaptive $(75.85\%)$ &  $0.5456$ & $[0.5066, 0.5820]$ & $-0.0085$ & $[-0.0574, 0.0403]$\\
genetic circuit & never &  NaN & NaN & NaN & NaN\\
genetic toggle & each jump &  $5.5083$ & $[5.0655, 5.9181]$ & $0.3701$ & $[-0.1219, 0.8622]$\\
genetic toggle & adaptive $(2.20\%)$ &  $5.5307$ & $[5.0344, 5.9860]$ & $0.2406
$ & $[-0.2541, 0.7353]$\\
genetic toggle & never &  $5.4175$ & $[4.9015, 5.8886]$ & $-0.8134$ & $[-1.2930, -0.3339]$\\
SEIR & each jump &  $2.0853$ & $[1.9323, 2.2279]$ & $ -0.0176$ & $[-0.2043, 0.1691]$\\
SEIR & adaptive $(9.91\%)$ &  $2.3726$ & $[2.2102, 2.5245]$ & $0.1238$ & $[-0.0884, 0.3359]$\\
SEIR & never &  $2.0823$ & $[1.9572, 2.2003]$ & $-0.0029$ & $[-0.1893, 0.1835]$\\
linear propensity $c_2$ & each jump & $0.4099$ & $[0.3850, 0.4334]$ & $0.0233$ & $ [-0.0134, 0.0599]$\\
linear propensity $c_2$& adaptive $(23.00\%)$ & $  0.4237$ & $[0.3977, 0.4481]$ & $-0.0295$ & $[-0.0674, 0.0083]$\\
linear propensity $c_2$ & never & $0.5380$ & $[0.5022, 0.5716]$ & $0.0628$ & $[0.0149, 0.1106]$\\
SEIR $\kappa$ & each jump &  $0.0377$ & $[0.0354, 0.0398]$ & $0.0002$ & $[-0.003, 0.004]$\\
SEIR $\kappa$ & adaptive $(20.61\%)$ &  $0.0351$ & $[0.0329, 0.0372]$ & $-0.003$ & $[-0.004, 0.0028]$\\
SEIR $\kappa$ & never &  $0.0371$ & $[0.0343, 0.0396]$ & $-0.0064$ & $[-0.0097, -0.003]$\\

\hline
\end{tabular}
\end{table*}

\section{Conclusions and future work} \label{sec-conclusions}
We presented a novel filtering algorithm to compute 
the conditional probability mass function
\[
\pi(t,x) = P\{X(t)=x \, | \, Y(s)=y(s), 0 \leq s \leq t\},
\]
in the context of reaction networks. Here $Y(t)=y(t)$ is the vector copy number of a subset of species that are assumed to be observed exactly in continuous time and $X(t)$ is the vector copy number of the remaining unobserved species. We also showed how this algorithm 
can be adapted for the purposes of Bayesian parameter estimation based on exact partial state observation. Furthermore, we also showed how the state $X(t_0)$ at time $t_0$ can be estimated based on observations up to a later time $T$ where $0 \leq t_0 \leq T$. 

The filtering algorithm involves a weighted Monte Carlo method and a resampling strategy needs to be employed. We explored some possibilities for resampling at the observed jump times $t_k$. Our investigations in this regards were numerical and based on relatively simple examples, and hence not exhaustive. An investigation of adaptive resampling based on some theoretical analysis is the subject of future work. 

While we presented an intuitive derivation of the filtering equations, our derivation is not mathematically rigorous. 
\cite{confortola2013filtering} provides a rigorous derivation in the context of finite state Markov processes which in the case of reaction networks correspond to systems where species conservation relations limit the species counts to be bounded. A rigorous derivation is certainly possible for the case of unbounded species copy numbers provided certain integrability or moment bound conditions hold. In this context conditions in Refs.~ \cite{rathinam2015Quarterly, gupta2014scalable, engblom2014} will be relevant.  

As is well known, many intracellular reaction networks may have some species in greater abundance and a discrete state model can be tedious to simulate one event at a time. 
Tau-leap methods \cite{gillespie2001approximate} as well as model reduction  approaches have been proposed for efficient simulation of such systems. See \cite{gillespie2007stochastic} and references therein for several methods as well as \cite{ball2006asymptotic,kang2013separation, hepp2015adaptive, ganguly2015jump, li2007analysis,  anderson2011error, rathinam2016convergence} for rigorous  mathematical analysis of methods. These same considerations could be applied to the filtering method proposed here to develop reduced order models 
and tau-leap simulations. 

Our assumption of exact (noiseless) observation of some species may appear unrealistic. However, as mentioned in 
the introduction, most observation noise may be modeled 
via extra reactions and extra species such as photons. 
If the photon counts are very large, the same considerations of reduced order models or tau-leaping mentioned above apply. 
Less realistic is the assumption of continuous in time observations. In reality, observations are recorded 
in discrete time snapshots. If the frequency of the snapshots is very high, then the theory of continuous in time observations provides a good approximation. If the frequency is low, then this is not the case.   
Future work will involve the case where 
observations of some species are made at certain time snapshots as well as the limiting behavior as the 
time snapshots increase in frequency. Either way, the theory and the algorithm discussed in this paper provides a baseline for the exploration of observations at discrete time snapshots.

\vspace{1em}

{\bf Acknowledgments:}
We thank Ankit Gupta for introducing us to the rich and subtle topic of stochastic filtering.

\appendix
\section{Derivation of the evolution equations for $\pi(t,x)$}
We note that \cite{confortola2013filtering} provides a rigorous derivation
of the evolution equation for the conditional probability $\pi(t,z)$ when the state space is finite
and the exact observation is of the form $y=h(z)$ where $h$ is a function of
the state space. The derivation in Ref.~ \cite{confortola2013filtering} may not
be easily accessible to applied scientists who may not be familiar with the
language of stochastic analysis. Moreover, our filtering equations (while in agreement with Ref.~ \cite{confortola2013filtering}) 
are somewhat simpler in appearance since in our case   
$h$ corresponds to the projection onto the last $n_2$ components of the state and also due to the
structure of the reaction network. 
The derivation shown here is more intuitive to follow (at the expense of some
rigor) and results in
equations consistent with \cite{confortola2013filtering}. Moreover, we shall
not make the assumption that the state space is finite. We believe that the
rigorous derivation in Ref.~ \cite{confortola2013filtering} can be extended to
infinite state space under reasonable assumptions, but such an endeavor
is beyond the scope of this paper. 

We start with a discretization of the time interval $[0,\infty)$ 
by a mesh
\[
\{\ell h \, | \ell = 0,1,\dots\}
\]
of spacing $h$. We consider $(X(\ell h),Y(\ell h))$ as a discrete time Markov
chain. If the observed trajectory of $Y(t)$ is given by $y(t)$, then
the observations on the mesh points will be given by
\[
\bar{y}_\ell = y(\ell h), \quad \ell=0,1,\dots
\]
We may use the filtering equations for a partially observed
discrete time Markov chain derived in \cite{fristedt2007filtering}.
For the discrete time Markov chain $(X(\ell h),Y(\ell h))$ (for $\ell=0,1,\dots$) where $Y$ is observed 
exactly, the conditional probability
\[
\pi_\ell(x) = P(X(\ell h) =x \, | \, Y(j h) = \bar{y}_{j} \, j=0,\dots,n),
\]
is shown in \cite{fristedt2007filtering} to satisfy 
\begin{equation}\label{eq-disc-pi}
  \pi_\ell(x') = \frac{\sum_x \kappa((x,\bar{y}_{\ell-1}),(x',\bar{y}_\ell)) \pi_{\ell-1}(x)}{\sum_{x_\ell}
    \sum_{x_{\ell-1}} \kappa((x_{\ell-1},\bar{y}_{\ell-1}),(x_\ell,\bar{y}_\ell)) \pi_{\ell-1}(x_{\ell-1})},
\end{equation}
where
\[
\kappa((x,y),(x',y')) = P(X(\ell h)=x', Y(\ell h)=y' \, | \,
X((\ell-1)h)=x,Y((\ell-1)h)=y).
\]
From the infinitesimal characteristics of continuous time
Markov chains, as $h$ approaches $0+$, 
\begin{equation}\label{eq-kappa-inf}
\begin{aligned}    
\kappa((x,y),(x',y')) &= a_j(x,y) h + o(h) \quad (x',y')=(x,y)+\nu_j,\\
\kappa((x,y),(x,y)) &= 1 - h \sum_{j=1}^m a_j(x,y) + o(h),\\
\kappa((x,y),(x',y')) &= o(h) \quad \text{ otherwise.}
\end{aligned}
\end{equation}
Let's consider two adjacent mesh points $t$ and $t+h$.  There are two
possibilities; there are no jumps of $y$ on the interval $(t,t+h]$ or
there are jumps. Again from the infinitesimal characteristics of continuous time
Markov chains, as $h$ approaches $0+$, there is either no jump or one jump
during $(t,t+h]$. We approximate $\pi(t,x)$ by $\pi_{\ell-1}(x)$
and $\pi(t+h,x)$ by $\pi_\ell(x)$. For the case when there is no jump during
$(t,t+h]$, if we suppose $t_k \leq t < t+h < t_{k+1}$, then
\[
\bar{y}_\ell=\bar{y}_{\ell-1} = y(t_k).
\]
Using \eqref{eq-disc-pi} and
\eqref{eq-kappa-inf}, we obtain
$\pi(t+h,x)$ as a ratio where
the numerator is
\[
\sum_{j \in \sU} \pi(t,x-\nu_j') a_j(x-\nu_j',y(t_k)) h +
\pi(t,x) - \pi(t,x) \sum_{j=1}^m a_j(x,y(t_k)) h + o(h)
\]
and the denominator is
\[
\sum_{\tilde{x}}
  \left(1 - \sum_{j \in \sO} a_j(\tilde{x},y(t_k)) h\right) \pi(t,\tilde{x})+o(h).
\]
From the above, we may obtain an expression for
$(\pi(t+h,x)-\pi(t,x))/h$ which upon taking limit as $h \to 0+$ yields 
\eqref{eq-pi-deriv}. The second case is when $t < t_k < t+h$ and in this case
\[
\bar{y}_{\ell-1} = y(t_{k-1}) \neq y(t_{k}) = \bar{y}_{\ell}.  
\]
Using \eqref{eq-disc-pi} and
\eqref{eq-kappa-inf}, we obtain
\[
\pi(t+h,x) = \frac{\sum_{l \in \sO_k}a_l(x-\nu'_l,y(t_{k-1})) \,
  \pi(t,x-\nu'_l) \, h + o(h)}{\sum_{\tilde{x}} \sum_{l \in \sO_k}
  a_l(\tilde{x}-\nu'_l,y(t_{k-1})) \, \pi(t,\tilde{x}) \, h + o(h)}.
\]
Noting that $\pi(t+h,x) \to \pi(t_k,x)$  and $\pi(t,x) \to
\pi(t_k-,x)$ as $h \to 0+$, we obtain  
\eqref{eq-pi-jump}. 

We note that for a rigorous treatment one needs to use the language of measure theory, since unlike in the discrete time Markov chain case, in the continuous time case we are conditioning on a zero probability event of observing $Y(s)=y(s)$ for $0 \leq s \leq t$. Moreover, a rigorous and mathematically ``cleaner'' treatment involves working with the integral representation of the differential equation with jumps. The derivation provided here, we hope, provides the ``essence'' of the idea. 

\section{Unnormalized and normalized filtering equations}\label{append-rho-pi}
We show that if $\rho(t,x)$ satisfies the unnormalized filtering equations, 
then $\pi(t,x)=\rho(t,x)/\sum_{\tilde{x}}{\rho}(t,\tilde{x})$. We assume the existence and uniqueness of solutions of both the unnormalized and normalized filtering equations. 

To that end, suppose $\rho$ solves the unnormalized filtering equations \eqref{eq-rho-deriv} and \eqref{eq-rho-jump}, and let $\tilde{\pi}(t,x)=\rho(t,x)/\sum_{\tilde{x}}{\rho}(t,\tilde{x})$. It is adequate to show that $\tilde{\pi}(t,x)$ satisfies the filtering equations \eqref{eq-pi-deriv} and \eqref{eq-pi-jump}. 

In between jump times, that is, for $t_k \leq t < t_{k+1}$, $\rho$ satisfies
\begin{equation*}
\begin{aligned}  
\rho'(t,x) &= \sum_{j \in \sU} \rho(t,x-\nu_j')\,a_j(x-\nu_j',y(t_k)) - \sum_{j \in
  \sU} \rho(t,x)\,a_j(x,y(t_k))\\
&- \rho(t,x)\, a^{\sO}(x,y(t_k)) \quad \forall x \in \posint^{n_1}.
\end{aligned}
\end{equation*}
Then
\begin{equation*}
\begin{split}
    \tilde{\pi}'(t,x) = \frac{ \rho'(t,x)} {\sum_{\tilde{x}}{\rho}(t,\tilde{x}) }- \rho(t,x) \frac{\sum_{\tilde{x}}{\rho'}(t,\tilde{x})}{(\sum_{\tilde{x}}{\rho}(t,\tilde{x}))^2}.
\end{split}
\end{equation*}
The first term may be written as
\begin{equation*}
    \frac{ \rho'(t,x)} {\sum_{\tilde{x}}{\rho}(t,\tilde{x}) } = \sum_{j \in \sU} \tilde{\pi}(t,x-\nu_j')\,a_j(x-\nu_j',y(t_k)) - \sum_{j \in
  \sU} \tilde{\pi}(t,x)\,a_j(x,y(t_k)) - \tilde{\pi}(t,x)\, a^{\sO}(x,y(t_k)).
\end{equation*}
The second term can be written as
\begin{equation*}
    \rho(t,x) \frac{\sum_{\tilde{x}}{\rho'}(t,\tilde{x})}{(\sum_{\tilde{x}}{\rho}(t,\tilde{x}))^2} = \pi(t,x) \frac{\sum_{\tilde{x}}{\rho'}(t,\tilde{x})}{\sum_{\tilde{x}}{\rho}(t,\tilde{x})}
\end{equation*}
Note that $\sum_{\tilde{x}}\left(\sum_{j \in \sU} \rho(t,\tilde{x}-\nu_j')\,a_j(\tilde{x}-\nu_j',y(t_k)) - \sum_{j \in
  \sU} \rho(t,\tilde{x})\,a_j(\tilde{x},y(t_k)) \right) = 0$, and hence
\begin{equation*}
\begin{split}
    \sum_{\tilde{x}}{\rho'}(t,\tilde{x}) &=  \sum_{\tilde{x}}\left(\sum_{j \in \sU} \rho(t,\tilde{x}-\nu_j')\,a_j(\tilde{x}-\nu_j',y(t_k)) - \sum_{j \in
  \sU} \rho(t,\tilde{x})\,a_j(\tilde{x},y(t_k)) - \rho(t,\tilde{x})\, a^{\sO}(\tilde{x},y(t_k))\right)\\
 &= -\sum_{\tilde{x}}\rho(t,\tilde{x})\, a^{\sO}(\tilde{x},y(t_k))
\end{split}
\end{equation*}
Hence $\tilde{\pi}$ satisfies 
\eqref{eq-pi-deriv}. 

For $k=1,2,\dots$ at jump times $t_k$, $\rho(t,x)$ jumps according to 
\begin{equation*}
\rho(t_k,x) = \frac{1}{|\sO_k|}\sum_{j \in \sO_k} a_j(x-\nu_j',y(t_{k-1})) \, \rho(t_k-,x-\nu_j') \quad x \in \posint^{n_1}.      
\end{equation*}
Since $\tilde{\pi}(t,x)=\rho(t,x)/\sum_{\tilde{x}}{\rho}(t,\tilde{x})$, we have
\begin{equation*}
\tilde{\pi}(t_k,x) = \frac{\sum_{l \in \sO_k} a_l(x-\nu_l',y(t_{k-1}))\,
  \tilde{\pi}(t_k-,x-\nu_l')}{\sum_{\tilde{x}}\sum_{l \in \sO_k}
  a_l(\tilde{x},y(t_{k-1}))\, \tilde{\pi}(t_k-,\tilde{x})} \quad \forall x \in \posint^{n_1}
\end{equation*}
which shows that $\tilde{\pi}(t,x)$ satisfies  \eqref{eq-pi-jump} at jump times $t_k$.

\nocite{*}
\bibliography{StateParaEst}
\bibliographystyle{plain}

\end{document}